\documentclass[english, aps,pra,twocolumn, superscriptaddress]{revtex4-1}
\usepackage[T1]{fontenc}
\usepackage[latin9]{inputenc}
\usepackage{mathtools}
\usepackage{amsmath}
\usepackage{amsthm}
\usepackage{amssymb}
\usepackage{graphicx}
\usepackage{array}

\makeatletter

\theoremstyle{definition}
\newtheorem{defn}{\protect\definitionname}
\theoremstyle{definition}
 \newtheorem{example}{\protect\examplename}

\usepackage{braket}
\usepackage{bbm}

\allowdisplaybreaks

\newcommand{\id}{\mathbbm{1}}

\let\oldsqrt\sqrt
\def\sqrt{\mathpalette\DHLhksqrt}
\def\DHLhksqrt#1#2{
\setbox0=\hbox{$#1\oldsqrt{#2\,}$}\dimen0=\ht0
\advance\dimen0-0.2\ht0
\setbox2=\hbox{\vrule height\ht0 depth -\dimen0}
{\box0\lower0.4pt\box2}}

\let\originalleft\left
\let\originalright\right
\renewcommand{\left}{\mathopen{}\mathclose\bgroup\originalleft}
\renewcommand{\right}{\aftergroup\egroup\originalright}

\renewcommand\bra[1]{{\langle{#1}|}}
\makeatletter
\renewcommand\ket[1]{
  \@ifnextchar\bra{\k@t{#1}\!}{\k@t{#1}}
}
\newcommand\k@t[1]{{|{#1}\rangle}}
\makeatother

\usepackage{bm}
\usepackage[unicode=true, breaklinks=false, pdfborder={0 0 1}, backref=false, colorlinks=true, linkcolor=blue, urlcolor=blue, citecolor=blue,bookmarksopen=true]{hyperref}

\newcommand{\Ab}[1]{ \left| #1 \, \right|}

\DeclareMathOperator{\Tr}{Tr}

\makeatletter
\def\th@newremark{\th@remark\thm@headfont{\bfseries}}
\makeatletter

\theoremstyle{newremark}
\newtheorem{observation}{Observation}

\makeatother

\usepackage{babel}
\providecommand{\definitionname}{Definition}
\providecommand{\examplename}{Example}

\theoremstyle{plain}
\newtheorem{thm}{\protect\theoremname}
\theoremstyle{plain}
\newtheorem{cor}{\protect\corollaryname}
\theoremstyle{plain}
\newtheorem{prop}{\protect\propositionname}
\providecommand{\corollaryname}{Corollary}
\providecommand{\propositionname}{Proposition}
\providecommand{\theoremname}{Theorem}

\begin{document}

\title{Blurred quantum Darwinism across quantum reference frames}

\author{Thao P. Le}

\email{thao.le.16@ucl.ac.uk}
\affiliation{Department of Physics and Astronomy, University College London, Gower Street, London WC1E 6BT}

\author{Piotr Mironowicz}

\affiliation{Department of Algorithms and System Modeling, Faculty of Electronics, Telecommunications and Informatics, Gda\'{n}sk University of Technology}
\affiliation{National Quantum Information Centre, University of Gda\'{n}sk, Wita Stwosza 57, 80-308 Gda\'{n}sk, Poland}
\affiliation{International Centre for Theory of Quantum Technologies, University of Gdansk, Wita Stwosza 63, 80-308 Gdansk, Poland}

\author{Pawe\l{} Horodecki}

\affiliation{National Quantum Information Centre, University of Gda\'{n}sk, Wita Stwosza 57, 80-308 Gda\'{n}sk, Poland}
\affiliation{International Centre for Theory of Quantum Technologies, University of Gdansk, Wita Stwosza 63, 80-308 Gdansk, Poland}

\date{\today}

\begin{abstract}
Quantum Darwinism describes objectivity of quantum systems via their correlations with their environment\textemdash information that hypothetical observers can recover by measuring the environments. However, observations are done with respect to a frame of reference. Here, we take the formalism of {[}Giacomini, Castro-Ruiz, \& Brukner. \emph{Nat Commun} \textbf{10}, 494 (2019){]}, and consider the repercussions on objectivity when changing quantum reference frames. We find that objectivity depends on non-degenerative relative separations, conditional state localisation, and environment macro-fractions. There is different objective information in different reference frames due to the interchangeability of entanglement and coherence, and of statistical mixing and classical correlations. As such, \textit{objectivity is subjective} across quantum reference frames.
\end{abstract}

\maketitle

\section{Introduction}

The emergence of the classical world from the underlying quantum mechanics remains a fundamental riddle. Quantum Darwinism is one particular approach that describes the emergence of \emph{objectivity} through the spread of information \citep{Zurek2003}. A system state $S$ is objective (or inter-subjective \citep{Mironowicz2017}) when many independent observers can determine the state of $S$ independently, without perturbing it, and arrive at the same result \citep{Ollivier2004,Horodecki2015}.

Quantum Darwinism can be seen as an extension of the decoherence theory. As systems interact with their surrounding environments, decoherence theory describes how quantum systems lose their coherence and decohere into a preferred pointer basis \citep{Joos1985,Schlosshauer2005,Schlosshauer2007}. The environment is not unchanged through this process\textemdash the system becomes correlated with the environment. Quantum Darwinism occurs if the information about the system has been proliferated into multiple fractions of the environment, such that many observers can access independent environments and gain equivalent information about the system. Objective states can be described either with \.Zurek's quantum Darwinism \citep{Zurek2009}, strong quantum Darwinism \citep{Le2019}, or spectrum broadcast structure \citep{Horodecki2015}. The emergence of these states have been studied extensively (for example, recent works include Refs.~\citep{Mironowicz2018,Roszak2019,Campbell2019,Le2019b,Chen2019,Unden2019,Lorenzo2019,Milazzo2019,Oliveira2019,Garcia-Perez2019,Tuziemski2019,Roszak2019a, Colafranceschi2020,Lorenzo2020,Qi2020}).

A key component of quantum Darwinism is the measurement performed by observers\textemdash which in physics, is done relative to some reference frame. However, in works thus far, one implicitly \emph{assumes} that all observers share the same classical frame. However, is the objectivity still consistent if observers do not share the same frame?

While classical reference frames are well established, there are numerous different proposals for describing quantum reference frames that focus on different aspects (for example, \citep{Bartlett2007,Gour2008,Palmer2014,Smith2016,Popescu2018,Giacomini2019,Loveridge2018}). In this paper, we apply the framework of \citet{Giacomini2019}, in which quantum reference frames are associated with a physical quantum state and vice-versa.

We examine objective states in different quantum reference frames. The method of \citet{Giacomini2019} allows us to move to the reference frame associated with any particular environment state which in turn is associated with the hypothetical observer frame. Entanglement and coherence have become interchangeable frame-dependent properties; as are statistical mixing and classical correlations. Such correlations are an intrinsic part of quantum Darwinism, hence, in general, objectivity \emph{does not} remain the same in different quantum reference frames. However, there \emph{are} certain conditions in which objectivity \emph{is} consistent, and conditions in which some kind of objectivity exists. To clearly show this, we consider static particles, such that changing quantum reference frames requires only changes in relative position, and we use the clear state structure afforded by spectrum broadcasting \citep{Horodecki2015}.

We show that, if all system and environment positions are exactly localised and randomly distributed (say, due random noise) then objectivity is consistent in all frames. We demonstrate that non-matching relative positions between all states is key to this consistency.

However, by allowing the system and environments to have a non-zero, continuous spread, objectivity distorts and blurs when changing quantum reference frames. The internal statistical mixedness and coherences of the environment states now play a crucial role in distributing new correlations. We find that the distinguishability of the other environment states depends on an interplay of relative distance separations and relative spreads; and how large macro-fractions of environments may be required to recover objectivity.

Finally, we analyse scenarios with a system interacting with environments to show how objectivity can arise dynamically, and to show how these factors---coherences, spectrum broadcasting, mixedness, and state separation---affect the level of objectivity in different frames.

This paper is structured as follows. In Sec. \ref{sec:Preliminaries}, we describe the frameworks of spectrum broadcast structure and quantum reference frames. In Sec. \ref{sec:Perfect-localisation-and-consistent-objectivity}, we depict some states that have consistent objectivity in all relevant quantum reference frames. In Sec. \ref{sec:Continuous-spread-and-blurred-objectivity} we examine objective states with a Gaussian-like spread. We describe the distortion of objectivity, and investigate the requirements for environment-state distinguishability that is a necessary component of quantum Darwinism. In Sec. \ref{sec:perfect_objectivity_theorems}, we prove the precise conditions for perfect objectivity in all quantum reference frames. In Sec. \ref{sec:Numerical-simulations}, we numerical investigate a fully coherent model involving a dynamic interaction between a system and two environments. We conclude in Sec. \ref{sec:Conclusion}.

\section{Preliminaries\label{sec:Preliminaries}}

\subsection{Spectrum broadcast structure\label{subsec:Quantum-Darwinism}}

In quantum Darwinism, we consider a central system $S$ that interacts and becomes correlated with its surrounding environment $E$. Typically, only a fragment, $F\subset E$ of the environment is measured and evaluated against the conditions for objectivity\textemdash as the full pure system-environment will retain coherences and entanglement under a global unitary evolution \citep{Roszak2019}. There are number of different frameworks that describe the properties of an objective state \citep{Zurek2009,Le2019,Horodecki2015}, each corresponding to slightly different strengths of objectivity. In this paper, we are focused on spectrum broadcast structure \citep{Horodecki2015}, because it has a clear state structure that allows us to explicitly calculate how the state changes under quantum reference frame transformations. Note that from here, when we speak of ``environment'', we refer to the \emph{observed} environment.
\begin{defn}
\textbf{Spectrum broadcast structure (SBS)} \citep{Horodecki2015}. \label{def:SBS}
A system-environment has spectrum broadcast structure when the joint state can be written as
\begin{align}
\rho_{SE} & =\sum_ip_i\ket{i}\bra{i}_S\otimes\rho_{E_1|i}\otimes\cdots\otimes\rho_{E_N|i},\label{eq:SBS}
\end{align}
where $\left\{ \ket{i}_S\right\} $ is the pointer basis, $p_i$ are probabilities, and all states $\rho_{E_k|i}$ are perfectly distinguishable, \emph{i.e.} $\Tr \left( \rho_{E_k|i}\rho_{E_k|j} \right) = 0\,\forall\,i\neq j,$ for each observed environment $E_k$.
\end{defn}
These states have zero discord~\cite{Ollivier2001} between the system and environments, feature maximal classical correlations between the system and environments, and satisfy strong independence (see Definition~\ref{defn:strongIndependence}).
All states with spectrum broadcast structure are objective, though not all objective states have spectrum broadcast structure \citep{Horodecki2015, Le2019}.

\subsection{Quantum reference frames\label{subsec:Quantum-reference-frames}}

As we noted, there is a number of different prescriptions for reference frames and quantum information (e.g. Refs. \citep{Bartlett2007,Gour2008,Palmer2014,Smith2016}). In this paper, we apply the framework of \citet{Giacomini2019}, which is inherently relational.

We consider the system and environments to be static (\emph{i.e.} without momentum) and distributed across space. Thus, a general reference frame transformation, $\hat{S}_{\text{position}}^{\left(C\rightarrow A\right)}$, is defined here as position only, as follows \citep{Giacomini2019}:
\begin{multline}
\hat{S}_{\text{position}}^{\left(C\rightarrow A\right)}\int dx_A dx_B\Psi\left(x_A,x_B\right)\ket{x_A}_A\ket{x_B}_B\\
=\int dq_B dq_C\Psi\left(-q_C,q_B-q_C\right)\ket{q_B}_B\ket{q_C}_C,
\end{multline}
\emph{i.e.} there is a coordinate transformation,  $x_A\rightarrow-q_C$, $x_B\rightarrow q_B-q_C$. We will always start in an implicit laboratory reference frame $(C)$, and move to the quantum reference frames centered on a particular quantum state. 

For our purposes,  SBS  is inherently mixed. Hence, if the initial state $\rho_{SE_1\cdots E_N}^{(C)}$ in the $(C)$ reference frame (laboratory frame) is
\begin{widetext}
\begin{equation}
\rho_{SE_1\cdots E_N}^{(C)}=\int dx_Sdx_S^{\prime} \left(\prod_{i=1}^{N}\int dx_{E_i}dx_{E_i}^{\prime} \right)\rho\left(x_S,x_{E_1},\ldots,x_{E_N},x_S^{\prime},x_{E_1}^{\prime},\ldots,x_{E_N}^{\prime}\right)\ket{x_S}\bra{x_S^{\prime}}_S\otimes\bigotimes_{j=1}^{N}\ket{x_{E_j}}\bra{x_{E_j}^{\prime}}_{E_j},
\end{equation}
then the transformation to the environment $E_1$ reference frame (without loss of generality) is:
\begin{align}
\label{eq:transform}
\rho_{SCE_2 \cdots E_N}^{\left(E_1\right)}  & =\int dq_S dq_S^{\prime}\int dq_Cdq_C^{\prime}\left(\prod_{i=2}^{N}\int dq_{E_i}dq_{E_i}^{\prime} \right) \ket{q_S}\bra{q_S^{\prime}}_S\otimes\ket{q_C}\bra{q_C^{\prime}}_C\otimes\bigotimes_{j=2}^{N}\ket{q_{E_j}}\bra{q_{E_j}^{\prime}}_{E_j}\nonumber \\
 & \qquad \times\rho\left(q_S-q_C,-q_C,q_{E_2}-q_C,\ldots,q_{E_N}-q_C,q_S^{\prime}-q_C^{\prime},-q_C^{\prime},q_{E_2}^{\prime}-q_C^{\prime},\ldots,q_{E_N}^{\prime}-q_C^{\prime}\right).
\end{align}

\end{widetext}

Entanglement and coherences in the position basis are quantum reference frame dependent \citep{Giacomini2019}. Furthermore, statistical (incoherent) mixtures and classical correlations are also frame dependent. Given that objectivity is built up from correlations between system and environment, and given that the environment states can contain coherences and statistical mixture, changing reference frames can have a serious effect on the objectivity of the system.

\section{Perfect localisation and objectivity in all quantum reference frames \label{sec:Perfect-localisation-and-consistent-objectivity}}

We consider a system $S$ and collection of environments $\left\{ E_i\right\} _{i=1,\ldots,N}$, such that they are objective in the laboratory frame $C$, and in particular have spectrum broadcast structure. The system and environments are located in a one-dimensional, continuous space, with positions $x_{X}$, $X=S,E_i$. In the idealised situation, these positions are perfectly localised, \emph{i.e.} existing at isolated points in space, and this allows us to gain insight into one of factors that contribute to consistent objectivity in all quantum reference frames\textemdash \emph{non-degenerative relative positions.}

We begin with section  \ref{subsec:GHZ-like-objective-states}, where we first examine the simplest, illustrative situation where the objective states have GHZ-like structure. In section \ref{subsec:Perfectly-localised-SBS}, we consider general perfectly localised objective SBS states.

\subsection{GHZ-like objective states\label{subsec:GHZ-like-objective-states}}

One of the simplest objective states possible is the reduced Greenberger\textendash Horne\textendash Zeilinger state (GHZ state); it is simpler yet again if its elements are incoherent in the position basis as follows:
\begin{equation}
\rho_{SE_1\cdots E_N}^{(C)}=\sum_ip_i\ket{x_i^S}\bra{x_i^S}_S\otimes\bigotimes_{j=1}^{N}\ket{x_i^{E_j}}\bra{x_i^{E_j}}_{E_j},\label{eq:GHZ_ISBS_state}
\end{equation}
which is objective provided that all $\left\{ x_i^S\right\}_i$, $\{ x_i^{E_j}\}_i$ are distinct. The implicit laboratory reference frame $(C)$ is perfectly localised and product with the system and environments. The objective information is characterised by the probability distribution $\left\{ p_i\right\}_i$.

In the frame of any of the environments\textemdash we take $E_1$ without loss of generality\textemdash the joint state now involves the laboratory $C$ as one of its subsystems, and now $E_1$ is implicit, perfectly localised and product with all other subsystems:
\begin{align}
\rho_{SCE_2\cdots E_N}^{\left(E_1\right)} & =\sum_ip_i\ket{x_i^S-x_i^{E_1}}\bra{x_i^S-x_i^{E_1}}_S\otimes\ket{-x_i^{E_1}}\bra{-x_i^{E_1}}_C\nonumber \\
 & \phantom{=}\otimes\bigotimes_{j=2}^{N}\ket{x_i^{E_j}-x_i^{E_1}}\bra{x_i^{E_j}-x_i^{E_1}}_{E_j}.\label{eq:GHZ_ISBS_E1frame}
\end{align}
In order for this to still be objective, \emph{and} with the same information as in the laboratory frame $C$, $\{ p_i\} _i$, we require that all $\{ x_i^S-x_i^{E_1}\} _i$ are distinct, and all $\{ x_i^{E_j}-x_i^{E_1}\}_i$ are distinct\textemdash that is, these terms are \emph{non-matching} or \emph{non-degenerate}. 

The majority of states of the form Eq.~(\ref{eq:GHZ_ISBS_state}) remain consistently objective in all quantum reference frames, in the following sense: If all the various positions $\{ x_i^S\}_i$, $\{ x_i^{E_1}\}_i$, etc. are \emph{randomly chosen} from a continuous interval, for example with probability mass function $f_{\text{uni}}\left(x\right)=1$, $x\in\left[0,1\right]$, then the probability that any two are equal is zero: $\mathbb{P}\left(x_i=x_j\right)=0$, due to the nature of discrete sampled numbers from uncountably infinite interval. Hence, any randomly drawn $\left\{ x_i^X\right\} _{i,X}$, $X=S,E_1,\ldots,E_N$ will produce an objective state for Eq.~(\ref{eq:GHZ_ISBS_state}). By the same argument, the probability that any relative separations $\{ x_i^S-x_i^{E_1}\}_i$ are equal is zero: $\mathbb{P}(x_i^S-x_i^{E_1}=x_j^S-x_j^{E_1})=0$, and hence all the terms in the system-environment state in any quantum reference frame, Eq.~(\ref{eq:GHZ_ISBS_E1frame}), are distinct and hence remains objective with the same spectrum probabilities $\left\{ p_i\right\}_i$.

Randomly sampled positions of the system and environment describe disorganised and noisy scenarios and models. However, solid state materials and lattices can have a rigid structure and hence potentially degenerate distances between state positions. In these situations, SBS and objectivity may become trivial in certain quantum reference frames.
\begin{example}
\label{eg:degenerate_localised}Consider the typical reduced GHZ state, where $x_i=i$ for $i=0,1$:
\begin{equation}
\rho_{SE_1\cdots E_N}^{(C)}= p_0 \ket{0}\bra{0}^{\otimes N+1}+p_1 \ket{1}\bra{1}^{\otimes N+1}.
\end{equation}
In the quantum reference frame of environment $E_1$, the state has the form
\begin{align}
\rho_{SCE_2\cdots E_N}^{\left(E_1\right)}= & \ket{0\cdots0}\bra{0\cdots0}_{SE_2\cdots E_N} \nonumber\\
 & \otimes\left(p_0\ket{0}\bra{0}_C+p_1\ket{-1}\bra{-1}_C\right).
\end{align}
The system and the remaining environments are trivially ``objective'' and uncorrelated. Meanwhile, the old information about the system has been shifted into the quantum system of the laboratory reference frame $C$.
\end{example}

\begin{observation}
If all positions are perfectly localised, the \emph{non-degeneracy} of the relative positions of the system and environments is crucial in ensuring the consistent objectivity in all quantum reference frames. If some of the relative distances between positions are not distinct, then the corresponding states can become non-distinguishable and thus degrade the original objectivity.
\end{observation}

In Appendix \ref{app:Objectivity-with-continuous}, we consider GHZ-like states with continuous objective probabilities, leading to an analogous requirement of non-degeneracy (in particular, continued injectivity of the functions mapping the continuous positions of the system and environment).

If instead the various positions $\{ x_i^X\}_{i,X} $ are picked uniformly from a finite set of $N$ positions, then the probability of two being the same is $\mathbb{P}\left(x_i=x_j\right)=1/N$. This goes to zero as $N\rightarrow\infty$. This situation can correspond to the case when there is a finite precision of a measurement device, and where any spread in the positions is much smaller than the device precision. In Sec.~\ref{sec:Continuous-spread-and-blurred-objectivity}, we will consider when there \emph{is} an inherent spread in the position, and in Sec.~\ref{sec:Numerical-simulations}, the positions of the system and environment are limit to a finite set. But firstly, in the following subsection, we consider general coherent\textemdash albeit still localised\textemdash objective states with spectrum broadcast structure.

\subsection{Perfectly localised spectrum broadcast states and new objectivity\label{subsec:Perfectly-localised-SBS}}

States with the SBS form typically contain coherences and mixtures in the conditional environment states. Under transformations of quantum reference frames, these can turn into global correlations. Combined with perfect localisation, we show how this produces a new, more complex objective information in different frames.

In general, a perfectly localised objective state with the SBS can be written as 
\begin{align}
\rho_{SE_1\cdots E_N}^{(C)} & =\sum_i p_i\ket{\psi_i^S}\bra{\psi_i^S}_S\otimes\bigotimes_{j=1}^{N}\rho_{E_j|i},
\end{align}
where we have general coherent states:
\begin{align}
\ket{\psi_i^S} & =\sum_k q_{k,i}\ket{x_{k|i}^S}_S, \\
\rho_{E_j|i} & =\sum_{k_j}t_{k_j,i}\ket{\varphi_{i,k_j}^{E_j}}\bra{\varphi_{i,k_j}^{E_j}}_{E_j}, \\
\ket{\varphi_{i,k_j}^{E_j}} & =\sum_{a_{ij}}r_{a_{ij},i,j,k_j}\ket{x_{a_{ij},k_j|i}^{E_j}}_{E_j}.
\end{align}
Objectivity requires that these states are orthogonal: $\braket{\psi_i^S|\psi_{i^{\prime}}^S}=0\,\forall\,i\neq i^{\prime}$ and $\braket{\varphi_{i,k_j}^{E_j}|\varphi_{i^{\prime},k_j^{\prime}}^{E_j}}=0\,\forall\,\left(i,k_j\right)\neq\left(i^{\prime},k_j^{\prime}\right)$. It is sufficient (though not necessary) if we let all the values $\left\{ x_{k|i}\right\} _{i,k}$, $\{ x_{a_{ij},k_j|i}^{E_j}\}_{i,k_j , a_{ij}}$ be randomly chosen numbers from a continuous interval, in which case the probability that any are equal is zero, hence all terms are orthogonal.

In the frame of environment $E_1$, the joint state has the following form: 
\begin{align}
 & \rho_{SCE_2\cdots E_N}^{\left(E_1\right)}\nonumber \\
 & =\sum_{i,k_1}p_i t_{k_1,i}\sum_{a_{i1},a_{i1}^{\prime}}r_{a_{i1},i,1,k_1}r_{a_{i1}^{\prime},i,1,k_1}^{*}\ket{\tilde{\psi}_{i,k_1,a_{i1}}^S}\bra{\tilde{\psi}_{i,k_1,a_{i1}^{\prime}}^S}_S\nonumber \\
 & \phantom{=}\otimes\ket{-x_{a_{i1},k_1|i}^{E_1}}\bra{-x_{a_{i1}^{\prime},k_1|i}^{E_1}}_C\nonumber \\
 & \phantom{=}\otimes\bigotimes_{j=2}^{N}\sum_{k_j}t_{k_j,i}\ket{\tilde{\varphi}_{i,k_j, a_{i1}}^{E_j}}\bra{\tilde{\varphi}_{i,k_j, a_{i1}^{\prime}}^{E_j}}_{E_j},
\end{align}
where 
\begin{align}
\ket{\tilde{\psi}_{i,k_1,a_{i1}}^S} & =\sum_k q_{k,i}\ket{x_{k|i}^S-x_{a_{i1},k_1|i}^{E_1}}_S\\
\ket{\tilde{\varphi}_{i,k_j, a_{i1}}^{E_j}} & =\sum_{a_{ij}} r_{a_{ij},i,j,k_j}\ket{x_{a_{ij},k_j|i}^{E_j}-x_{a_{i1},k_1|i}^{E_1}}.
\end{align}
Due to the coherences and statistical mixedness of the original environment $E_1$ state, there is now entanglement and correlations between the system and the environment in the $\left(E_1\right)$ frame. In particular, much of the entanglement is tied with the laboratory subsystem $C$\textemdash and to the indices $a_{i1}$ and $a_{i1}^{\prime}$ that came from the original $E_1$ state. Hence \emph{if} the positions $\left\{ x_{a_{i1},k_1|i}^{E_1}\right\} _{i,k_1, a_{i1}}$ are distinct, then we can trace out $C$ and remove the system-environment entanglement:
\begin{align}
\rho_{SE_2\cdots E_N}^{\left(E_1\right)} & =\sum_{i,k_1,a_{i1}}p_it_{k_1,i}\left|r_{a_{i1},i,1,k_1}\right|^2 \nonumber \\
 & \phantom{=}\times\ket{\tilde{\psi}_{i,k_1,a_{i1}}^S}\bra{\tilde{\psi}_{i,k_1,a_{i1}}^S}_S\otimes\bigotimes_{j=2}^{N}\tilde{\rho}_{E_j|i,k_1,a_{i1}},\\
\tilde{\rho}_{E_j|i,k_1,a_{i1}} & \coloneqq\sum_{k_j}t_{k_j,i}\ket{\tilde{\varphi}_{i,k_j,  a_{i1}}^{E_j}}\bra{\tilde{\varphi}_{i,k_j, a_{i1}}^{E_j}}_{E_j}.
\end{align}
From the assumption that all the $\left\{ x_{\cdots}\right\} _{\cdots}$ are randomly sampled from a continuous distribution, all the relative differences $\{ x_{k|i}^S-x_{a_{i1},k_1|i}^{E_1}\}_{k,i,k_1}$, $\left\{ x_{a_{ij},k_j|i}^{E_j}-x_{a_{i1},k_1|i}^{E_1}\right\} _{i,j,k_j,k_1}$ are unique, hence the conditional states of the system and the environments are perfectly distinguishable, and the reduced state $\rho_{SE_2\cdots E_N}^{\left(E_1\right)}$ has the SBS. However, the objective information is now encoded by the probabilities $\left\{ p_i t_{k_1,i}\left|r_{a_{i1},i,1,k_1}\right|^2\right\} _{i,k_1,a_{i1}}$. Although the original system information can still be recovered by taking the relevant marginal distribution, we see that in each different reference frame corresponding to environment $E_j$, we will have a different set of objective information. 

\begin{observation}
Coherences in the environment can create entanglement between the system, lab, and remainder environments. This can typically be ``removed'' by tracing out the laboratory subsystem.
\end{observation}

\begin{observation}
Incoherent mixedness in the environment creates new classical correlations between the system, lab, and remainder environments. This can lead to \emph{new} objective information, which includes the original information which can be recovered from the marginals by summing over terms associated with the environment.
\end{observation}

Hence, while entanglement and coherence are frame-dependent properties, it is equally relevant that incoherent mixedness and classical correlations are also frame-dependent. Only a very small class of objective states retain the same objectivity in different quantum reference frames: and more generally, the system objectivity transforms to a more complicated objectivity, of which the original system information is embedded within.

\section{Continuous spread and blurred objectivity\label{sec:Continuous-spread-and-blurred-objectivity}}

Thus far, we have shown how non-degeneracy of relative positions plays a crucial role in objectivity, when positions are perfectly localised. However, in general, systems and environments have a non-zero spread. In this section, we examine systems and environments with a continuous spread described by Gaussian distributions across space, characterised by mean ${\mu}$ and standard deviation ${\sigma}$. Objectivity becomes blurred and distinguishability reduces as states become ``smeared'' across space in different reference frames.

In Section \ref{subsec:P_error_and_fidelity}, we describe the error probability of distinguishing conditional states, and how that is bounded by the fidelity. This fidelity becomes our measure for a perceived objectivity. In Section \ref{subsec:Incoherent-unmixed-SBS}, we consider incoherent objective states, in which the conditional environment states are single Gaussians for simplicity, and in Section \ref{subsec:Coherent-SBS}, we consider general coherent objective states.

\subsection{Effective perceived objectivity and the fidelity of measurement\label{subsec:P_error_and_fidelity}}

One method to quantify compliance with the SBS is with a distance measure to the set of the SBS states. For example, some of us \citep{Mironowicz2017} have developed a computable tight bound $\eta\left[\rho_{SF}\right]$ on the trace distance (where $F$ denotes a subset of environment states). For a predefined basis, the system can be written as $\rho_{S}=\sum_i p_i\ket{i}\bra{i}+\sum_{i\neq j} p_{ij}\ket{i}\bra{j}$, and then SBS distance bound is:
\begin{align}
T^{\text{SBS}}\left(\rho_{SF}\right) & =\dfrac{1}{2}\min_{\rho_{SF}^{\text{SBS}}}\left\Vert \rho_{SF}-\rho_{SF}^{\text{SBS}}\right\Vert _{1}\leq\eta\left[\rho_{SF}\right]\label{eq:SBStracedistance}\\
\eta\left[\rho_{SF}\right] & \equiv\Gamma+\sum_{i\neq j}\sqrt{p_ip_{j}}\sum_{k=1}^{F}B\left(\rho_{E_k|i},\rho_{E_k|j}\right),\label{eq:eta}
\end{align}
where $\Gamma$ describes the coherence of the system relative to a predefined basis, \emph{i.e.} encoding the $\{p_{ij}\}_{i\neq j}$ terms, and
\begin{equation}
	\label{eq:B}
	B\left(\rho_i,\rho_j\right)=\left\Vert \sqrt{\rho_i}\sqrt{\rho_j}\right\Vert _1
\end{equation}
is the fidelity describing the distinguishability of the conditional environment states, and $\{p_i\}_i$ are probabilities of the system in the predefined basis. The above bound, however, implicitly assumes strong independence of the environments:
\begin{defn}
	\label{defn:strongIndependence}
\textbf{Strong independence} \citep{Horodecki2015}. Sub-environments $\left\{ E_k\right\} _k$ have strong independence relative to the system $S$ if their conditional mutual information is vanishing:
\begin{equation}
I\left(E_j:E_k|S\right)=0,\quad\forall j\neq k.
\end{equation}
\end{defn}
Unlike the work in \citep{Mironowicz2017}, strong independence is \emph{not} maintained in general when changing quantum reference frames. However, strong independence is not required for a more general objectivity \citep{Le2019}.

Here, we focus on the distinguishability of the conditional states. For an ensemble $\left\{ p_i,\rho_i\right\}_i $, and a set of measurement operators $\left\{ \Pi_i\right\}_i$, $\sum_i\Pi_i=\id$, to pick out $i$, the average probability of successful measurement is:
\begin{equation}
\mathbb{P}\left(\text{success}\right)=\sum_i p_i\Tr\left[\rho_i\Pi_i\right],
\end{equation}
and the average probability of failure is then $\mathbb{P}\left(\text{error}\right)=1-\mathbb{P}\left(\text{success}\right)$. The minimum error of distinguishing the states is bounded by the fidelity of the conditional states \citep{Barnum2002,Montanaro2008}:
\begin{equation}
\label{eq:errBound}
\sum_{i<j}p_i p_{j}\left\Vert \sqrt{\rho_i}\sqrt{\rho_{j}}\right\Vert _{1}^2 \leq\mathbb{P}\left(\text{error}\right)\leq\sum_{i\neq j}\sqrt{p_i p_{j}}\left\Vert \sqrt{\rho_i}\sqrt{\rho_{j}}\right\Vert _{1}.
\end{equation}

Hence, the fidelity $\left\Vert \sqrt{\rho_i}\sqrt{\rho_{j}}\right\Vert _{1}$ is the key term that we will be calculating in this section. We will also occasionally calculate the overlap between two states, 
\begin{equation}
L\left(\rho_i,\rho_{j}\right)=\Tr\left[\rho_i\rho_{j}\right],
\end{equation}
which gives a lower bound on the fidelity, $L\left(\rho_i,\rho_{j}\right)\leq\left\Vert \sqrt{\rho_i}\sqrt{\rho_{j}}\right\Vert _{1}^2 $.

For perfect objectivity, it is necessary (but not sufficient) for $\mathbb{P}\left(\text{error}\right)=0$, hence the lower-bound to $\mathbb{P}\left(\text{error}\right)$ in turn gives a minimum distance from objectivity. 

\subsection{Incoherent, unmixed objective states and blurred objectivity\label{subsec:Incoherent-unmixed-SBS}}

Consider the following incoherent objective state with the SBS,
\begin{equation}
\rho_{SE_1\cdots E_N}^{(C)}=\sum_ip_i\ket{x_i^S}\bra{x_i^S}_S\otimes\bigotimes_{j=1}^{N}\rho_{E_j|i},\label{eq:incoherent_cont_SBS}
\end{equation}
where the environment states are unmixed (in the sense of consisting of a single Gaussian-distributed state rather than a discrete sum of Gaussians):
\begin{align}
\rho_{E_j|i} & =\int dx_{E_j}f\left(x_{E_j}|\mu_{E_j|i},\sigma_{E_j|i}\right)\ket{x_{E_j}}\bra{x_{E_j}}.
\end{align}
We have defined the Gaussian (normal) probability density
\begin{equation}
f\left(x|\mu,\sigma\right)=\dfrac{1}{\sqrt{2\pi}\sigma}\exp\left[-\dfrac{1}{2}\left(\dfrac{x-\mu}{\sigma}\right)^2 \right].
\end{equation}
This allows us to focus on the effects of the Gaussian spread on the objectivity. From the very beginning, there is no perfect objectivity: the fidelity between two conditional environment states for $i,i^{\prime}$ is:
\begin{align}
\left\Vert \sqrt{\rho_{E_j|i}}\sqrt{\rho_{E_j|i^{\prime}}}\right\Vert _{1} & =\dfrac{\exp\left[-\dfrac{\left(\mu_{E_j|i}-\mu_{E_j|i^{\prime}}\right)^2 }{ 4 \left(\sigma_{E_j|i}^2 +\sigma_{E_j|i^{\prime}}^2 \right)}\right]}{\dfrac{\sqrt{\sigma_{E_j|i}^2 +\sigma_{E_j|i^{\prime}}^2 }}{\sqrt{2\sigma_{E_j|i}\sigma_{E_j|i^{\prime}}}}},\label{eq:incoh_SBS_E_fidelity}
\end{align}
which is always non-zero. As we impose that our original state in the laboratory frame is objective, this fidelity must be sufficiently small for all $i\neq i^{\prime}$. Hence, for any pair of $i\neq i^{\prime}$, we must either have $\mu_{E_j|i}-\mu_{E_j|i^{\prime}}\gg\sqrt{\sigma_{E_j|i}^2 +\sigma_{E_j|i^{\prime}}^2 }$, \emph{i.e.} the peak separations are larger than the standard deviation, or $\sigma_{E_j|i}\gg\sigma_{E_j|i^{\prime}}$ (or vice-versa), \emph{i.e.} one conditional state must have a larger spread than the others\textemdash this allows for the detection of the wide-spread distribution \emph{outside} the bulk to the sharper distribution. These two cases are depicted in Fig.~\ref{fig:separated_means_or_different_SDs}.

\begin{observation}
Objectivity requires the distinguishability of conditional states. If the conditional states are described with a Gaussian distribution, then the distinguishability requires a combination of sufficiently far separated peaks $\{\mu\}$, or otherwise sufficiently different spreads $\{\sigma\}$.
\end{observation}

\begin{figure}
\begin{centering}
\includegraphics[width=0.45\textwidth]{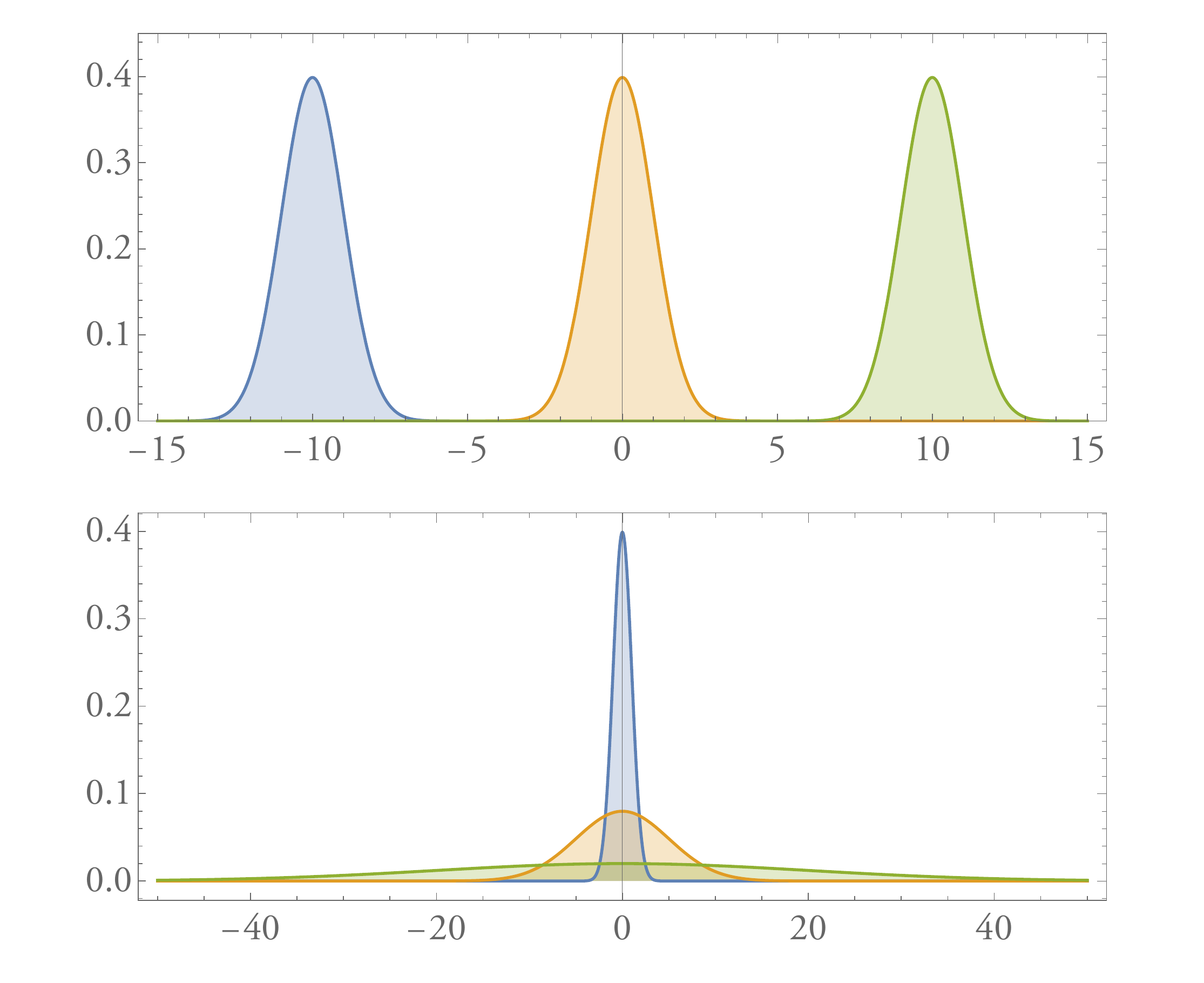}
\par\end{centering}
\caption{(Color on-line) Top: if the peaks for the states are separated much further than their standard deviations (here, $\Delta\mu=10\sigma$), then there is very little overlap and these states are distinguishable. Bottom: Alternatively, if the central peaks are the same or very close, varying greatly standard deviations (here, five times or more) allows for good distinguishability, as measurement at further locations will, with high probability, correspond to the wider distributions.\label{fig:separated_means_or_different_SDs}}
\end{figure}

Without loss of generality, we change to the quantum reference frame of the first environment $E_1$:
\begin{align}
 & \rho_{SC E_2\cdots E_N}^{\left(E_1\right)}\nonumber \\
 & =\sum_i p_i\int dq_Cf\left(-q_C|\mu_{E_1|i},\sigma_{E_1|i}\right)\nonumber \\
 & \phantom{=}\times\ket{x_i^S+q_C}\bra{x_i^S+q_C}_S\otimes\ket{q_C}\bra{q_C}\nonumber \\
 & \phantom{=}\otimes\bigotimes_{j=2}^{N}\int dq_{E_j}f\left(q_{E_j}-q_C|\mu_{E_j|i},\sigma_{E_j|i}\right)\ket{q_{E_j}}\bra{q_{E_j}}.\label{eq:incoh_SBS_state_E1frame}
\end{align}
In the new frame, the system is centered around $x_i^S-\mu_{E_1|i}$, the old laboratory $C$ is centered around $-\mu_{E_1|i}$, and the other environments have a complex distribution with a continuum of multiple peaks, at $q_C+\mu_{E_j|i}$, where $q_C$ is centered around $-\mu_{E_1|i}$. While the original system-objective information still exists, there are now extra classical correlations given across by $\int dq_C$. This continuum across $q_C$ means that we \emph{do not} have objectivity for the continuous distribution $\left\{ p_if\left(-q_C|\mu_{E_1|i},\sigma_{E_1|i}\right)\right\} _{i,q_C}$, as the states given by $q_C$ versus $q_C+\delta$ are not well distinguished. This is depicted in Fig.~\ref{fig:smooth_continuum_of_peaks}

\begin{figure}
\begin{centering}
\includegraphics[width=1\columnwidth]{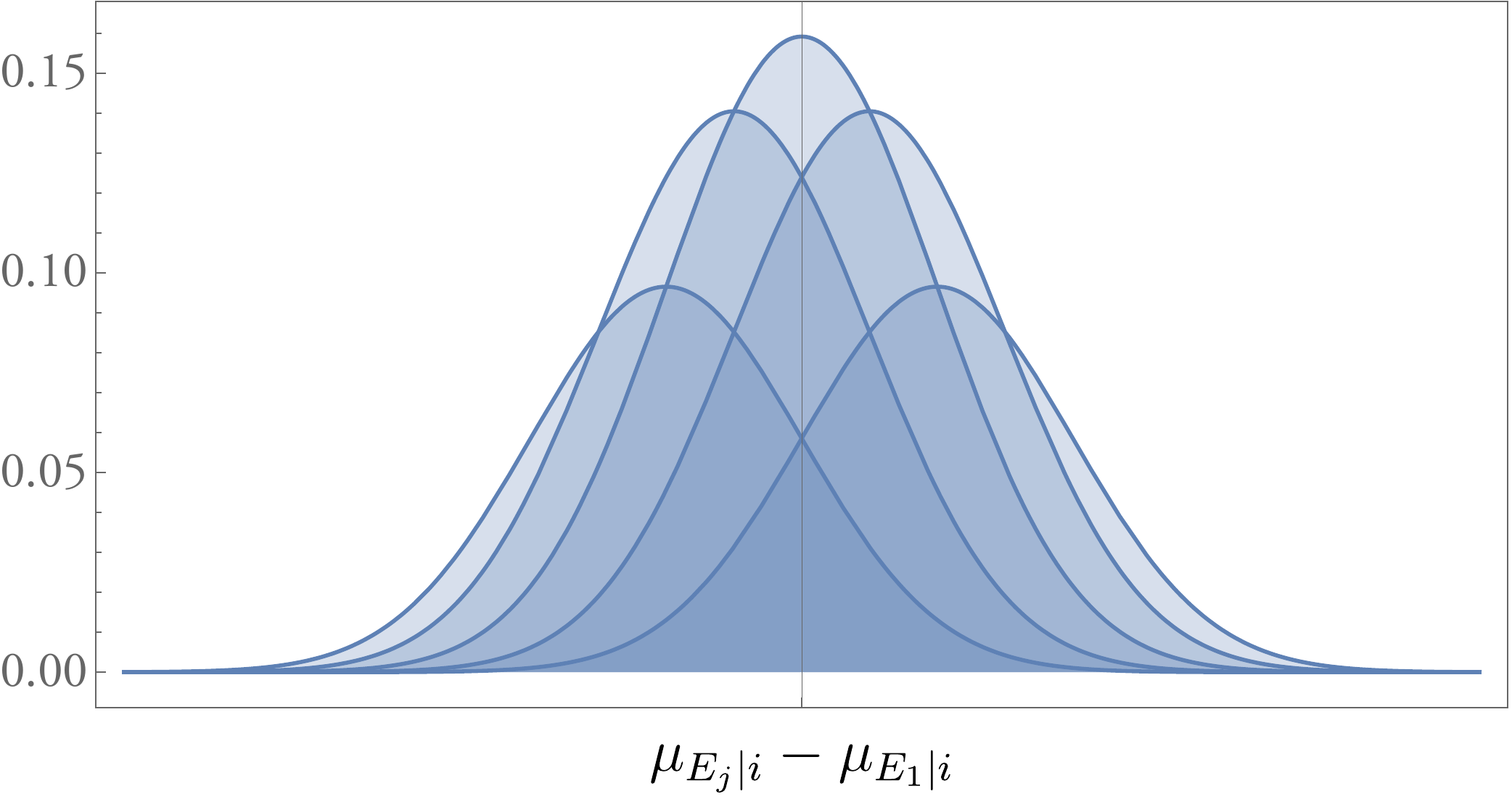}
\par\end{centering}
\caption{(Color on-line) In the frame of $\left(E_1\right)$, the original peaks of environment $E_j$, at $\mu_{E_j|i}$ are shifted by $q_C$, which ranges over the entire space, but with a Gaussian envelope centered at $-\mu_{E_1|i}$. Every curve corresponds to a different environment state conditioned on a different $q_C$ (for a fixed $i$). Different $q_C$ give curves that overlap a lot, and hence are not distinguishable. \label{fig:smooth_continuum_of_peaks}}
\end{figure}

Thus, the most immediate, and preferred, candidate for objectivity is the original information indexed by $i$. Firstly, the new conditional system states must be distinguishable. The local system state is $\rho_S^{\left(E_1\right)}=\sum_ip_i\rho_{S|i}^{\left(E_1\right)}$, where the conditional states are:
\begin{align}
\label{eq:condStates}
\rho_{S|i}^{\left(E_1\right)} & \coloneqq\int dq_Cf\left(-q_C|\mu_{E_1|i},\sigma_{E_1|i}\right)\ket{x_i^S+q_C}\bra{x_i^S+q_C}_S.
\end{align}
The fidelity of conditional system states is
\begin{align}
 & \left\Vert \sqrt{\rho_{S|i}^{\left(E_1\right)}}\sqrt{\rho_{S|i^{\prime}}^{\left(E_1\right)}}\right\Vert _{1}\nonumber \\
 & =\dfrac{\exp\left[-\dfrac{\left(x_i^S-\mu_{E_1|i}-x_{i^{\prime}}^S+\mu_{E_1|i^{\prime}}\right)^2 }{4\left(\sigma_{E_1|i}^2 +\sigma_{E_1|i^{\prime}}^2 \right)}\right]}{\left.\sqrt{\sigma_{E_1|i}^2 +\sigma_{E_1|i^{\prime}}^2 }\middle/\sqrt{2\sigma_{E_1|i}\sigma_{E_1|i^{\prime}}}\right.}.
\end{align}
Distinguishability requires a low fidelity, which occurs either if the shifted distances are non-degenerate with a sufficiently large separation, or if \emph{one} of $\sigma_{E_1|i}\gg\sigma_{E_1|i^{\prime}}$.

The reduced state on environment $E_j$ is $\rho_{E_j}^{\left(E_1\right)}=\sum_i p_i\rho_{E_j|i}^{\left(E_1\right)}$, with conditional states
\begin{align}
\rho_{E_j|i}^{\left(E_1\right)}\coloneqq & \int dq_C\int dq_{E_j}f\left(-q_C|\mu_{E_1|i},\sigma_{E_1|i}\right)\nonumber \\
 & \times f\left(q_{E_j}-q_C|\mu_{E_j|i},\sigma_{E_j|i}\right)\ket{q_{E_j}}\bra{q_{E_j}}.
\end{align}
The fidelity of the conditional states is:

\begin{widetext}
\begin{equation}
\left\Vert \sqrt{\rho_{E_j|i}^{\left(E_1\right)}}\sqrt{\rho_{E_j|i^{\prime}}^{\left(E_1\right)}}\right\Vert _{1}=\dfrac{\sqrt{2}\left[\left(\sigma_{E_1|i}^2 +\sigma_{E_1|i^{\prime}}^2 \right)\left(\sigma_{E_j|i}^2 +\sigma_{E_j|i^{\prime}}^2 \right)\right]^{1/4}}{\sqrt{\sigma_{E_1|i}^2 +\sigma_{E_1|i^{\prime}}^2 +\sigma_{E_j|i}^2 +\sigma_{E_j|i^{\prime}}^2 }}\exp\left[-\dfrac{\left(\mu_{E_1|i}-\mu_{E_1|i^{\prime}}-\mu_{E_j|i}+\mu_{E_j|i^{\prime}}\right)^2 }{4\left(\sigma_{E_1|i}^2 +\sigma_{E_1|i^{\prime}}^2 +\sigma_{E_j|i}^2 +\sigma_{E_j|i^{\prime}}^2 \right)}\right].\label{eq:SBS_incoh_env_fidelity}
\end{equation}

\end{widetext}

Once again, distinguishability requires low fidelity, which occurs if the shifted differences are very non-degenerate: $\mu_{E_1|i}-\mu_{E_1|i^{\prime}}-\mu_{E_j|i}+\mu_{E_j|i^{\prime}}\gg\sqrt{\sigma_{E_1|i}^2+\sigma_{E_1|i^{\prime}}^2+\sigma_{E_j|i}^2+\sigma_{E_j|i^{\prime}}^2}$, \emph{or} if at least one of the standard deviations $\sigma\in\left\{ \sigma_{E_1|i},\sigma_{E_1|i^{\prime}},\sigma_{E_j|i}\sigma_{E_j|i^{\prime}}\right\} $ is separated from the others by orders of magnitude, so that $\left\Vert \sqrt{\rho_{E_j|i}^{\left(E_1\right)}}\sqrt{\rho_{E_j|i^{\prime}}^{\left(E_1\right)}}\right\Vert _{1}\sim1/\sqrt{\sigma}\rightarrow0$ for $\sigma\rightarrow\infty$.

\begin{observation}
For the original information to remain objective in all frames, a necessary condition is good local distinguishability (local perceived objectivity). This requires a combination of very non-degenerate relative separations and very localised conditional states; or conditional spreads that vary by orders of magnitude.
\end{observation}

Suppose the conditional fidelity $\left\Vert \sqrt{\rho_{E_j|i}^{\left(E_1\right)}}\sqrt{\rho_{E_j|i^{\prime}}^{\left(E_1\right)}}\right\Vert _1$ is \emph{not} close to zero. In this case, we can take \emph{macro-fractions} in order to increase distinguishability. Suppose we have a fraction $F=\left\{ E_j\right\} _{j\in F}$. Then, the conditional fidelity is
\begin{align}
\left\Vert \sqrt{\rho_{F|i}^{\left(E_1\right)}}\sqrt{\rho_{F|i^{\prime}}^{\left(E_1\right)}}\right\Vert _{1} & =\prod_{j\in F}\left\Vert \sqrt{\rho_{E_j|i}^{\left(E_1\right)}}\sqrt{\rho_{E_j|i^{\prime}}^{\left(E_1\right)}}\right\Vert _{1}.
\end{align}
Provided $\left\Vert \sqrt{\rho_{E_j|i}^{\left(E_1\right)}}\sqrt{\rho_{E_j|i^{\prime}}^{\left(E_1\right)}}\right\Vert _{1}<1$, which is true provided that there is non-degeneracy in the relative positions, $\mu_{E_1|i}-\mu_{E_1|i^{\prime}}-\mu_{E_j|i}+\mu_{E_j|i^{\prime}}\neq0$, then the product of increasingly many of them takes $\left\Vert \sqrt{\rho_{F|i}^{\left(E_1\right)}}\sqrt{\rho_{F|i^{\prime}}^{\left(E_1\right)}}\right\Vert _{1}\rightarrow 0$.

\begin{observation}
Information becomes less distinguishable in different frames. Provided that there is non-degeneracy in the relative peak-positions, distinguishability can be achieved by taking a suitably large collection of sub-environments (macrofractions).
\end{observation}

In Fig.~\ref{fig:Amount-of-localisation_fragment_distinguishability}, we demonstrate the interplay between localisation and macrofraction size and their contribution to the distinguishability of two conditional states. 

\begin{figure}
\begin{centering}
\includegraphics[width=1\columnwidth]{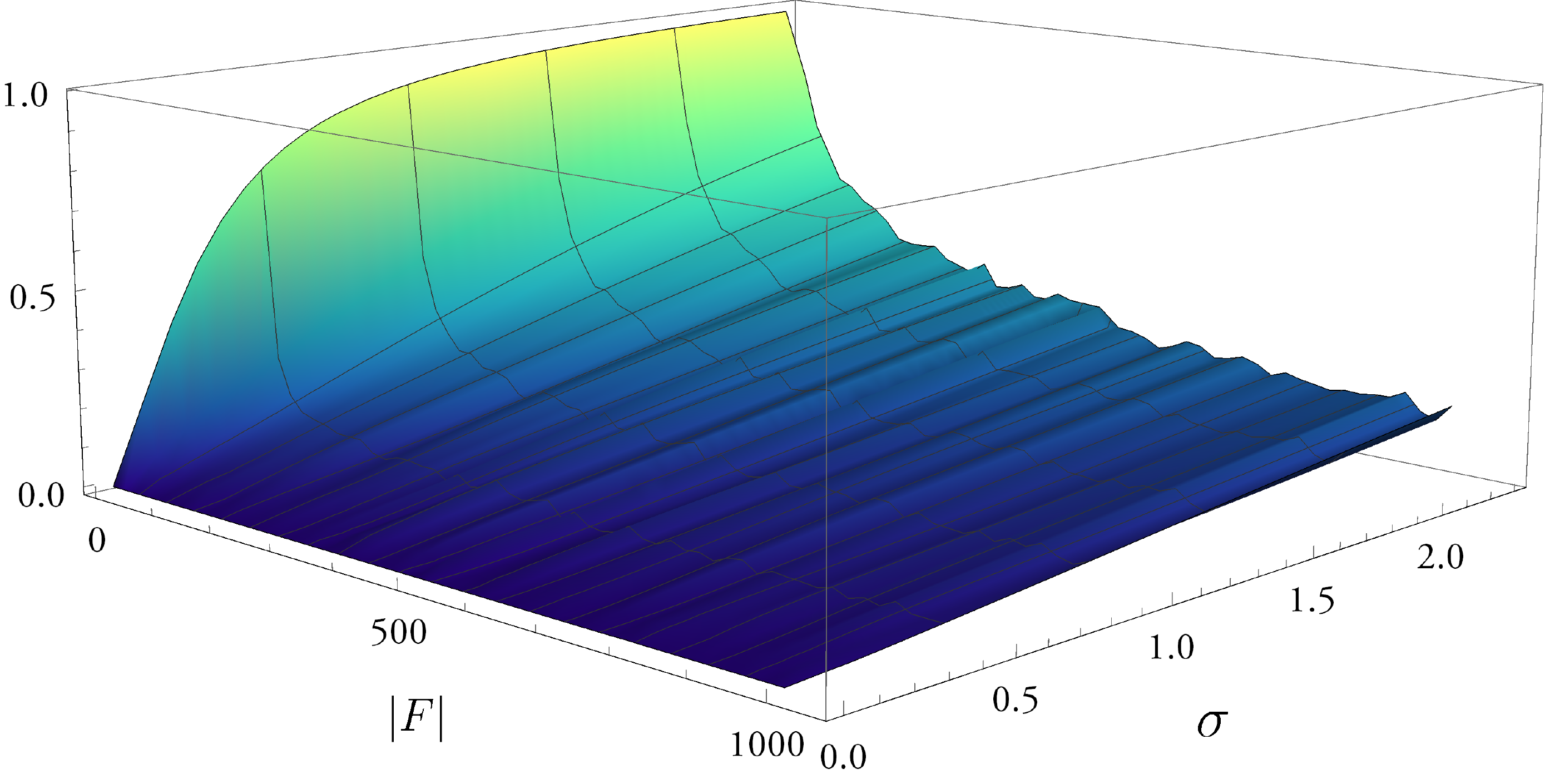}
\par\end{centering}
\caption{(Color on-line) Plot of the conditional state fidelity $\left\Vert \sqrt{\rho_{F|0}^{\left(E_1\right)}}\sqrt{\rho_{F|1}^{\left(E_1\right)}}\right\Vert _{1}$, versus the amount of localisation ($\sigma$ the same for all Gaussian states) and macrofraction size $|F|$, for the case of the state in Eq.~(\ref{eq:incoherent_cont_SBS}). Here, the peak positions $\left\{ \mu_{E_i|0},\mu_{E_i|1}\right\}_i$ are picked randomly from the interval $\left[-1,1\right]$ and the graph is averaged over $400$ collections of random samples. Sharp localisation $\sigma\rightarrow 0$, and large macrofractions $|F|$ lead to low conditional state fidelity and hence greater distinguishability.\label{fig:Amount-of-localisation_fragment_distinguishability}}
\end{figure}

In general, the conditional environment states can be mixed, e.g. $\rho_{E_1|i}=\sum_k q_k \rho_{E_1|i,k}$ from Eq.~(\ref{eq:incoherent_cont_SBS}) can be a mixture of distinguishable Gaussian states. As previously determined in Section \ref{sec:Perfect-localisation-and-consistent-objectivity}, this leads to a new objective information given by the distribution $\left\{ p_i q_k\right\} _{i,k}$, where the original information is recovered through the marginal obtained by summing over all values of $q_k$.

\subsection{Coherent objectivity states and the rise of new classical and quantum correlations\label{subsec:Coherent-SBS}}

In general, objective states have coherence. When moving to the reference frame of one of those environments, this coherence turns into entanglement between the other subsystems. Consider the following state, in which the system and environments are coherent relative to the position basis:
\begin{equation}
\rho_{SE_1\cdots E_N}^{(C)}=\sum_i p_i\ket{\psi_i^S}\bra{\psi_i^S}_S\otimes\bigotimes_{j=1}^{N}\rho_{E_j|i},
\end{equation}
where all the pure states are Gaussian wave-packets:
\begin{align}
\ket{\psi_i^S} & =\int dx_S f^{\frac{1}{2}}\left(x_S|\mu_{S|i},\sigma_{S|i}\right)\ket{x_S}_S\\
\rho_{E_j|i} & =\sum_{k_j}t_{k_j,i}\ket{\varphi_{i,k_j}^{E_j}}\bra{\varphi_{i,k_j}^{E_j}}_{E_j}\\
\ket{\varphi_{i,k_j}^{E_j}} & =\int dx_{E_j}f^{\frac{1}{2}}\left(x_{E_j}|\mu_{E_j|i,k_j},\sigma_{E_j|i,k_j}\right)\ket{x_{E_j}}_{E_j}.
\end{align}
Note that $f^{\frac{1}{2}}\left(\cdot\right)=\sqrt{f\left(\cdot\right)}$ is the square-root of a Gaussian (which may include a potential phase). In the reference frame of environment $E_1$, 
\begin{align}
 & \rho_{SCE_2\cdots E_N}^{\left(E_1\right)}\nonumber \\
 & =\sum_{i,k_1}p_it_{k_1,i}\int dq_S dq_S^{\prime}dq_Cdq_C^{\prime}f^{\frac{1}{2}}\left(q_S-q_C|\mu_{S|i},\sigma_{S|i}\right)\nonumber \\
 & \phantom{=}\times f^{\frac{1}{2}*}\left(q_S^{\prime}-q_C^{\prime}|\mu_{S|i},\sigma_{S|i}\right)f^{\frac{1}{2}}\left(-q_C|\mu_{E_1|i,k_1},\sigma_{E_1|i,k_1}\right)\nonumber \\
 & \phantom{=}\times f^{\frac{1}{2}*}\left(-q_C^{\prime}|\mu_{E_1|i,k_1},\sigma_{E_1|i,k_1}\right)\ket{q_S}\bra{q_S^{\prime}}_S\otimes\ket{q_C}\bra{q_C^{\prime}}_C\nonumber \\
 & \phantom{=}\otimes\bigotimes_{j=2}^{N}\sum_{k_j}t_{k_j,i}\int dq_{E_j}dq_{E_j}^{\prime}f^{\frac{1}{2}}\left(q_{E_j}-q_C|\mu_{E_j|i,k_j},\sigma_{E_j|i,k_j}\right)\nonumber \\
 & \qquad\quad f^{\frac{1}{2}*}\left(q_{E_j}^{\prime}-q_C^{\prime}|\mu_{E_j|i,k_j},\sigma_{E_j|i,k_j}\right)\ket{q_{E_j}}\bra{q_{E_j}^{\prime}}_{E_j}.
\end{align}
Coherence in the environment states (relative to the position basis in which we change reference frames) leads to an entanglement between the laboratory frame state and the system-environments. This entanglement can be removed by tracing out the laboratory state. The small changes in $q_C$ will not be distinguishable. Instead, the best candidate for the perceived objective information is $\left\{ p_it_{k_1,i}\right\} _{i,k_1}$---\emph{i.e.}, the original objectivity information mixed with the $E_1$ incoherent statistical mixedness that has now turned into classical correlations in the new frame as we have seen with previous examples.

The local system state is $\rho_S^{\left(E_1\right)}=\sum_{i,k_1}p_it_{k_1,i}\rho_{S|i,k_1}^{\left(E_1\right)}$, where the conditional states are:
\begin{align}
\rho_{S|i,k_1}^{\left(E_1\right)} & \coloneqq\int dq_Cf\left(-q_C|\mu_{E_1|i,k_1},\sigma_{E_1|i,k_1}\right)\nonumber \\
 & \qquad\left[\int dq_Sf^{\frac{1}{2}}\left(q_S-q_C|\mu_{S|i},\sigma_{S|i}\right)\ket{q_S}\right]\nonumber \\
 & \qquad\left[\int dq_S^{\prime}f^{\frac{1}{2}*}\left(q_S^{\prime}-q_C|\mu_{S|i},\sigma_{S|i}\right)\bra{q_S^{\prime}}_S\right].
\end{align}
The system is conditionally centered around $\mu_{S|i}-\mu_{E_1|i,k_1}$, with a spread of approximately $\sqrt{\sigma_{S|i}^2 +\sigma_{E_1|i,k_1}^2 }$.

\begin{observation}
In other reference frames, the conditional system states is typically no longer pure, but they can still be distinguishable. We can consider this a \emph{generalised objectivity}, in which the conditional system states are mixed (instead of conditionally pure) and perfectly distinguishable in the manner the environment states are.
\end{observation}

Heuristically, provided that these new peaks are sufficiently separated, or that different standard deviations separated by orders of magnitude, then the conditional states will be distinguishable. Since we cannot calculate the eigendecomposition for $\rho_{S|i,k_1}^{\left(E_1\right)}$ in general, we will instead calculate the overlap/linear fidelity, which is a lower bound to the fidelity: 
\begin{align}
 & \Tr\left[\rho_{S|i,k_1}^{\left(E_1\right)}\rho_{S|i^{\prime},k_1^{\prime}}^{\left(E_1\right)}\right]\nonumber \\
 & =\dfrac{2\sigma_{S|i}\sigma_{S|i^{\prime}}\exp\left[-\dfrac{\left(\mu_{E_1|i,k_1}-\mu_{S|i}-\mu_{E_1|i^{\prime},k_1^{\prime}}+\mu_{S|i^{\prime}}\right)^2 }{2\left(\sigma_{E_1|i,k_1}^2 +\sigma_{E_1|i^{\prime},k_1^{\prime}}^2 +\sigma_{S|i}^2 +\sigma_{S|i^{\prime}}^2 \right)}\right]}{\sqrt{\left(\sigma_{S|i}^2 +\sigma_{S|i^{\prime}}^2 \right)\left(\sigma_{E_1|i,k_1}^2 +\sigma_{E_1|i^{\prime},k_1^{\prime}}^2 +\sigma_{S|i}^2 +\sigma_{S|i^{\prime}}^2 \right)}}.
\end{align}
The linear fidelity is small when the relative differences are greater than the standard deviations, or if $\sigma_{E_1|i,k_1}$ are large compared to $\sigma_{S|i}$. 

Similarly, the environment states, $\rho_{E_j}^{\left(E_1\right)}=\sum_{i,k_1}p_it_{k_1,i}\rho_{E_j|i,k_1}^{\left(E_1\right)}$, have  conditional states
\begin{align}
\rho_{E_j|i,k_1}^{\left(E_1\right)} & \coloneqq \sum_{k_j}t_{k_j,i}\int dq_Cf\left(-q_C|\mu_{E_1|i,k_1},\sigma_{E_1|i,k_1}\right)\nonumber \\
 & \phantom{=}\left[\int dq_{E_j}f^{\frac{1}{2}}\left(q_{E_j}-q_C|\mu_{E_j|i,k_j},\sigma_{E_j|i,k_j}\right)\ket{q_{E_j}} \right]\nonumber \\
 & \phantom{=}\left[\int dq_{E_j}^{\prime}f^{\frac{1}{2}*}\left(q_{E_j}^{\prime}-q_C^{\prime}|\mu_{E_j|i,k_j},\sigma_{E_j|i,k_j}\right)\bra{q_{E_j}^{\prime}}_{E_j} \right].
\end{align}
We could calculate their linear fidelity (not shown here): provided they are separated in position, or if their standard deviations are very different, then the conditional environment states will be distinguishable.

When the environment states have coherence, the full system-environment state gains entanglement in other quantum reference frames. However, this entanglement can be decohered into classical correlations by tracing out the (transformed) laboratory system. Distinguishability requires the locations of the new peaks in the new reference frame to be sufficiently separated, or that the size of the spreads in the new reference frame be sufficiently different, and can be enhanced with macrofractions. Once distinguishable, the information $\left\{ p_i t_{k_1,i}\right\} _{i,k_1}$ can be recovered from the environments, and in turn the original system information. However, the $ t_{k_1,i}$ component is unique to the $(E_1)$ frame.

\begin{observation}
The system information $\left\{ p_i\right\}_i $ from the laboratory frame is unique, in that it is recoverable in all frames.
\end{observation}

Note though this is \emph{not} the same as saying that  $\left\{ p_i\right\}_i $  is \emph{objective} in all frames, as all the previous and following examples have shown.

\begin{observation}
All the information in the system-environment remains when changing reference frames. However, this information can become scrambled and prevent the system information $\left\{ p_i\right\}_i $ from being the only objective information in the new frames. Instead, the internal information of the new environment frame (\emph{i.e.} mixedness and coherence in the conditional states) produces new correlations that augment the original objective system information. Thus, to keep the exact same objective information, there should be as little internal conditional information in the environment as possible.
\end{observation}

\section{Precise conditions for Perfect Objectivity in all Quantum Reference Frames \label{sec:perfect_objectivity_theorems}}

In general, a discrete SBS state (i.e. containing countably many terms) can be written as follows: 
\begin{align}
\rho_{SE_{1}\cdots E_{N}}^{\left(C\right)} & =\sum_{i}p_{i}\ket{\psi_{i}^{S}}\bra{\psi_{i}^{S}}_{S}\otimes\bigotimes_{j=1}^{N}\rho_{E_{j}|i},\label{eq:general_discrete_SBS}\\
\braket{\psi_{i}^{S}|\psi_{i^{\prime}}^{S}} & =0,\quad\forall\,i\neq i^{\prime}\label{eq:general_discrete_condition_1}\\
\rho_{E_{j}|i}\rho_{E_{j}|i^{\prime}} & =0,\quad\forall\,i\neq i^{\prime},\forall\,j,\label{eq:general_discrete_condition_2}
\end{align}
where we have general coherent states on the system and mixed states on the environment that are perfectly distinguishable under different index $i$. We can write $\ket{\psi_{i}^{S}}$ and $\rho_{E_{j}|i}$ in general as:
\begin{align}
\ket{\psi_{i}^{S}} & =\sum_{x_{S}}\psi\left(x_{S}\Big|i\right)\ket{x_{S}}_{S},\\
\rho_{E_{j}|i} & =\sum_{x_{E_{j}},x_{E_{j}}^{\prime}}t\left(x_{E_{j}},x_{E_{j}}^{\prime}\Big|i,j\right)\ket{x_{E_{j}}}\bra{x_{E_{j}}^{\prime}}_{E_{j}}.
\end{align}
The objective information here is $\left\{ p_{i}\right\} _{i}$. However, as the cases above show, SBS states do not always remain SBS in different frames, and if they do, they will often have a \emph{different} objective information. In the following theorem, we give the particular SBS structure required for the same objective information in all relevant frames:
\begin{thm}
\label{thm:perfect_objectivity_discrete}A discrete SBS state $\rho_{SE_{1}E_{2}\cdots E_{N}}^{\left(C\right)}$ [Eq.~(\ref{eq:general_discrete_SBS})] is perfectly objective, with the same objective information $\left\{ p_{i}\right\} _{i}$, in all lab and environment reference frames if and only if it can be written in the following reduced form:
\begin{align}
\rho_{SE_{1}\cdots E_{N}}^{\left(C\right)} & =\sum_{i}p_{i}\ket{\psi_{i}^{S}}\bra{\psi_{i}^{S}}_{S}\otimes\bigotimes_{j=1}^{N}\ket{x_{E_{j}|i}}\bra{x_{E_{j}|i}}_{E_{j}},
\end{align}
and satisfying the perfect distinguishability conditions in the original lab frame:
\begin{align}
\braket{\psi_{i}^{S}|\psi_{i^{\prime}}^{S}} & =0,\quad\forall i\neq i^{\prime},\\
\braket{x_{E_{j}|i}|x_{E_{j}|i^{\prime}}} & =0,\quad\forall i\neq i^{\prime},\,\forall j,
\end{align}
and all environment frames:
\begin{align}
\braket{\tilde{\psi}_{i,j}^{S}|\tilde{\psi}_{i^{\prime},j}^{S}} & =0,\quad\forall i\neq i^{\prime},\,\forall j\label{eq:general_discrete_system_shifts}\\
\braket{x_{E_{j}|i}-x_{E_{k}|i}|x_{E_{j}|i^{\prime}}-x_{E_{k}|i^{\prime}}} & =0,\quad\forall i\neq i^{\prime},\,\forall j\neq k,\label{eq:general_discrete_environment_nondegen}
\end{align}
where $\ket{\tilde{\psi}_{i,j}^{S}}=\sum_{q_{S}}\psi\left(q_{S}+x_{E_{j}|i}|i\right)\ket{q_{S}}_{S}$.
\end{thm}
The proof is given in Appendix \ref{app:proof_theorem}: it proceeds by considering the general transformed state of the system-environment in frame $E_1$ (without loss of generality) and imposes that the system spectrum remains $\{p_i\}_i$ (which enforces the environment states $\rho_{E_{j}|i}$  conditioned on $i$ be pure) and that SBS is preserved (which gives the distinguishability conditions).

Thus, not only are the conditional environment states pure, they  must also non-degenerate separations [Eq.~(\ref{eq:general_discrete_environment_nondegen})]. This can be easily achieved by introducing randomness to the precise $\left\{ x_{E_{j}|i}\right\} $ terms.  An example of perfect objective states is given in Sec.~\ref{subsec:GHZ-like-objective-states}.

However, the orthogonality conditions Eq.~(\ref{eq:general_discrete_system_shifts}) for the system states $\left\{ \ket{\tilde{\psi}_{i,j}^{S}}\right\} _{i}$ in frame $E_{j}$ are much more nontrivial. It is possible that the shifts in the wavefunction from $\psi\left(x_{S}|i\right)\rightarrow\psi\left(x_{S}+x_{E_{j}|i}|i\right)$
can cause overlaps in the conditional system states in the new frame.
We depict this in Fig. \ref{fig:system_states_distinguishability}.

\begin{figure}
\begin{centering}
\includegraphics[width=0.9\columnwidth]{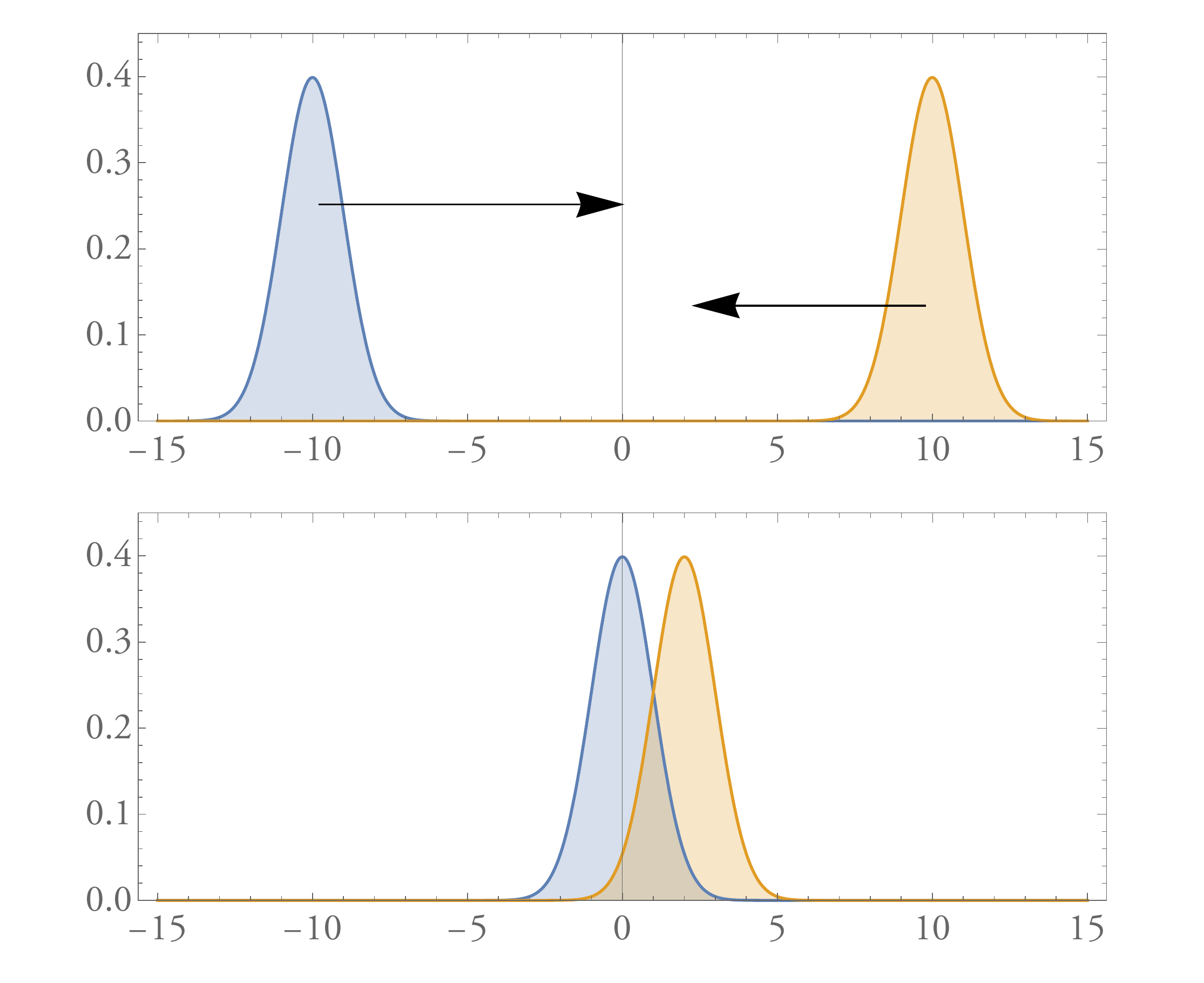}
\par\end{centering}
\caption{(Color on-line) Top: Curves representing the system conditional states in the original lab frame: they are separated and hence distinguishable. Bottom: Both curves are shifted by a \emph{different amount} as we move to an environment frame, yet the curves overlap and are no longer distinguishable. \label{fig:system_states_distinguishability}}
\end{figure}

Suppose we start of with a general continuous SBS state instead of a discrete one:
\begin{align}
\rho_{SE_{1}\cdots E_{N}}^{\left(C\right)} & =\int di\cdot p_{i}\ket{\psi_{i}^{S}}\bra{\psi_{i}^{S}}_{S}\otimes\bigotimes_{j=1}^{N}\rho_{E_{j}|i},\label{eq:continuous_state}\\
\ket{\psi_{i}^{S}} & =\int dx_{S}\cdot\psi\left(x_{S}\Big|i\right)\ket{x_{S}}_{S},\\
\rho_{E_{j}|i} & =\int dx_{E_{j}}dx_{E_{j}}^{\prime}\cdot t\left(x_{E_{j}},x_{E_{j}}^{\prime}\Big|i,j\right)\ket{x_{E_{j}}}\bra{x_{E_{j}}^{\prime}}_{E_{j}},
\end{align}
satisfying distinguishability conditions:
\begin{align}
\braket{\psi_{i}^{S}|\psi_{i^{\prime}}^{S}} & =0,\qquad\forall\,i\neq i^{\prime},\\
\rho_{E_{j}|i}\rho_{E_{j}|i^{\prime}} & =0,\qquad\forall\,i\neq i^{\prime},\forall\,j.
\end{align}

\begin{cor}
A continuous SBS state is perfectly objective, with the same (possibly continuous) objective information $\left\{ p_{i}\right\} _{i}$, in all lab and environment reference frames if and only if it satisfies the \textbf{same }state structure as given in Theorem \ref{thm:perfect_objectivity_discrete} (up to a continuous $i$), that is, with form
\begin{align}
\rho_{SE_{1}\cdots E_{N}}^{\left(C\right)} & =\int di\cdot p_{i}\ket{\psi_{i}^{S}}\bra{\psi_{i}^{S}}_{S}\otimes\bigotimes_{j=1}^{N}\ket{x_{E_{j}|i}}\bra{x_{E_{j}|i}}_{E_{j}},
\end{align}
and all other conditions given in Theorem \ref{thm:perfect_objectivity_discrete}.
\end{cor}
\begin{proof}
Take the continuous limit on the sums on the system and environment states from the discrete SBS state, $\sum_{x_{S}}\rightarrow\int dx_{S}$, $\sum_{x_{E_{j}}}\rightarrow\int dx_{E_{j}}$ and follow the same proof: perfect objectivity collapses those sums to discrete states and all other conditions follow.
\end{proof}
\begin{table}
\caption{Summary of the minimal specialised SBS state structure required for perfect objectivity in other quantum reference frames (QRFs), aside from detailed distinguishability conditions.\label{tab:Summary-of-the-environment-structure}}
\vspace{0.5em}
\centering{}%
\begin{tabular}{|>{\centering}p{4cm}|>{\centering}p{4cm}|}
\hline
Objectivity type & State structure requirement \tabularnewline
\hline 
Objective in all QRFs, with the same classical information $\left\{ p_{i}\right\} $  & All environment conditional states are pure in $x$ basis and localised (Thm. \ref{thm:perfect_objectivity_discrete},  Sec. \ref{subsec:GHZ-like-objective-states})\tabularnewline
\hline 
Objective in all QRFs, but with different objective information & All environment conditional states are incoherent and mixed in $x$ basis (Prop. \ref{prop:A-discrete-SBS_with_new_information}, Sec. \ref{subsec:Perfectly-localised-SBS}, Sec. \ref{subsec:Incoherent-unmixed-SBS})\tabularnewline
\hline 
A \textbf{reduced} state is objective in all QRFs, with different objective information & Environment conditional states can be coherent (Cor. \ref{corr:reduced_state_objective}, Sec. \ref{subsec:Coherent-SBS})\tabularnewline
\hline 
\end{tabular}
\end{table}

If we relax the requirement that same objective information appears, then we can relax the conditionally pure environment states to incoherent environment states in the $x$ basis:
\begin{prop}
\label{prop:A-discrete-SBS_with_new_information}A discrete SBS state $\rho_{SE_{1}E_{2}\cdots E_{N}}^{\left(C\right)}$ of the following form can be perfectly objective in all frames ($C$, $E_{j}$), albeit with \emph{different} objective information:
\begin{align}
\rho_{SE_{1}\cdots E_{N}}^{\left(C\right)} & =\sum_{i}p_{i}\ket{\psi_{S|i}}\bra{\psi_{S|i}}_{S}\nonumber \\
 & \quad\otimes\bigotimes_{j=1}^{N}\sum_{x_{E_{j}}}t\left(x_{E_{j}}\Big|i,j\right)\ket{x_{E_{j}}}\bra{x_{E_{j}}}_{E_{j}},
\end{align}
provided it satisfies the perfect distinguishability conditions in the original lab frame:
\begin{align}
\braket{\psi_{S|i}|\psi_{S|i^{\prime}}} & =0,\quad\forall i\neq i^{\prime}\\
\rho_{E_{j}|i}\rho_{E_{j}|i^{\prime}} & =0,\quad\forall i\neq i^{\prime}
\end{align}
 and all environment frames:
\begin{align}
\braket{\tilde{\psi}_{(i,q_{C|i})}^{\left(E_{j}\right)}|\tilde{\psi}_{(i^{\prime},q_{C|i^{\prime}}^{\prime})}^{\left(E_{j}\right)}} & =0,\quad\forall (i,q_{C|i}) \neq (i^{\prime},q_{C|i^{\prime}}^{\prime}),\\
\rho_{E_{k}|(i,q_{C|i})}^{\left(E_{j}\right)}\rho_{E_{k}|(i^{\prime},q_{C|i^{\prime}}^{\prime})}^{\left(E_{j}\right)} & =0,\quad\forall(i,q_{C|i})\neq(i^{\prime},q_{C|i^{\prime}}^{\prime}),\forall k\neq j,
\end{align}
where 
\begin{align}
\ket{\tilde{\psi}_{\left(i,q_{C|i}\right)}^{\left(E_{j}\right)}} & =\sum_{q_{S}}\psi\left(q_{S}-q_{C|i}\Big|i\right)\ket{q_{S}}_{S}\\
\rho_{E_{k}|\left(i,q_{C|i}\right)}^{\left(E_{j}\right)} & =\sum_{q_{E_{j}}}t\left(q_{E_{j}}-q_{C|i}\Big|i,j\right)\ket{q_{E_{j}}}\bra{q_{E_{j}}}_{E_{j}}.
\end{align}
Note that the values $q_{C|i}=q_{C|i}^{\left(E_{1}\right)}$ can take depends on the index $i$ and the original states on $E_{1}$.
\end{prop}
The proof is given in Appendix \ref{app:proof_proposition}. 
Note that this proposition is not an \emph{if-and-only-if}: we have chosen that the objective information in frame $E_{1}$, for example, is $\left\{ p_{i}t(-q_{C|i},-q_{C|i}\Big|i,j=1)\right\} _{(i,q_{C|i})}$, leading to the conditions in the proposition. An example of Proposition \ref{prop:A-discrete-SBS_with_new_information} is depicted in Fig. \ref{fig:prop1_successful_example}.

In the continuous case, this proposition will hold only up to some error, e.g. $\braket{\tilde{\psi}_{\left(i,q_{C|i}\right)}^{\left(E_{j}\right)}|\tilde{\psi}_{\left(i^{\prime},q_{C|i^{\prime}}^{\prime}\right)}^{\left(E_{j}\right)}}=\delta>0$. With continuous environments states\textemdash even if they are incoherent\textemdash will result in a reduced distinguishability as given in Fig.~\ref{fig:smooth_continuum_of_peaks}.

\begin{figure*}
\includegraphics[width=0.85\textwidth]{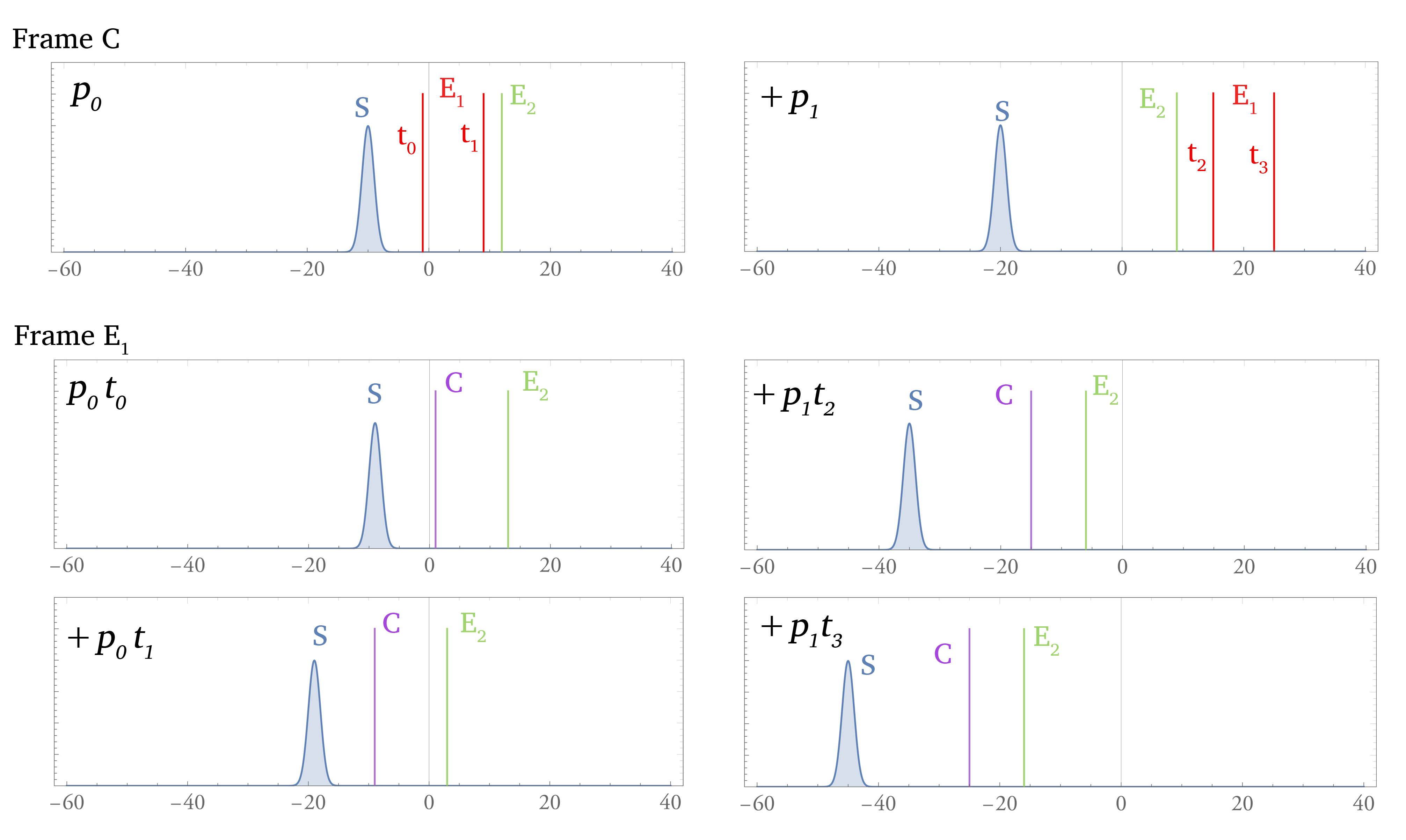}

\caption{(Color on-line) Example of an objective SBS state described by Proposition~\ref{prop:A-discrete-SBS_with_new_information}. Top: depiction of the state in frame $C$. The original objective information is $\{p_0,p_1\}$. The environment $E_{1}$ has mixed, incoherent states conditioned on $i=0$ ($p_{0}$, top left), and $i=1$ ($p_{1}$, top right). Bottom:  When moving into the quantum reference frame of $E_{1}$, the system states have been shifted such that they still remain distinguishable. The new objective information is $\{p_0 t_0, p_0 t_1, p_1 t_2, p_1 t_3\}$.
\label{fig:prop1_successful_example}}
\end{figure*}

\begin{cor}
\label{corr:reduced_state_objective}Consider the general discrete SBS state $\rho_{SE_{1}E_{2}\cdots E_{N}}^{\left(C\right)}$ from Eq.~(\ref{eq:general_discrete_SBS}). If we allow for \emph{partial trace} in other frames, the reduced state in those frames may be objective, provided the reduced state satisfies distinguishability conditions.
\end{cor}
That is, we have no particular state-structure restrictions from the general SBS state [Eq.~(\ref{eq:general_discrete_SBS})] (unlike in Theorem~\ref{thm:perfect_objectivity_discrete} and Proposition~\ref{prop:A-discrete-SBS_with_new_information}).

If we move from frame $C$ to frame $E_{k}$, then trace out the subsystem $C$, $\Tr_{C}\left[\rho_{SCE_{1}\ldots E_{N}}^{\left(E_{k}\right)}\right]$, then the reduced state could be objective (provided it satisfies the nontrivial distinguishability conditions). We did not need to restrict the conditional environment states to be localised or incoherent\textemdash sometimes, simply tracing out a subsystem can give an SBS state. For example, the GHZ state $(\ket{0000} + \ket{1111})/\sqrt{2}$ is entangled and not SBS, but tracing out a single subsystem and we are left with an SBS state, $(\ket{000}\bra{000} + \ket{111}\bra{111})/2$.

This corollary implies that the conditional environment states \emph{can} have coherences, and also shows how important, intricate, and nontrivial the distinguishability conditions are to the objectivity of a state, and emphasises our focus on the indistinguishability of conditional states in other parts of this paper.
\begin{proof}
From the general discrete SBS state in the frame of $E_{1}$, Eq.~(\ref{eq:general_system_state_E1_frame}), we trace out the $C$-subsystem:
\begin{align}
\rho_{SE_{2}\cdots E_{N}}^{\left(E_{1}\right)} & =\sum_{i,q_{C|i}}p_{i}t\left(-q_{C|i},-q_{C|i}\Big|i,j=1\right)\nonumber \\
 & \quad\times\ket{\tilde{\psi}^S_{\left(i,q_{C|i}\right)}}\bra{\tilde{\psi}^S_{\left(i,q_{C|i}\right)}}_{S}\otimes\bigotimes_{j=2}^{N}\rho_{E_{j}|i,q_{C|i}}^{\left(E_{1}\right)}\\
\ket{\tilde{\psi}^S_{\left(i,q_{C|i}\right)}} & =\sum_{q_{S}}\psi\left(q_{S}-q_{C|i}\Big|i\right)\ket{q_{S}}_{S}\\
\rho_{E_{j}|i,q_{C|i}}^{\left(E_{1}\right)} & =\sum_{q_{E_{j}},q_{E_{j}}^{\prime}}t\left(q_{E_{j}}-q_{C|i},q_{E_{j}}^{\prime}-q_{C|i}\Big|i,j\right) \nonumber\\
& \qquad\qquad \qquad\times\ket{q_{E_{j}}}\bra{q_{E_{j}}^{\prime}}_{E_{j}}.
\end{align}
By choosing $\left\{ p_{i}t\left(-q_{C|i},-q_{C|i}\Big|i,j=1\right)\right\} _{\left(i,q_{C|i}\right)}$
as the new objective information, this reduced state will have SBS provided it satisfies distinguishability conditions:
\begin{align}
\braket{\tilde{\psi}^S_{\left(i,q_{C|i}\right)}|\tilde{\psi}^S_{\left(i^{\prime},q_{C|i}^{\prime}\right)}} & =0,\quad\forall\left(i,q_{C|i}\right)\neq\left(i^{\prime},q_{C|i}^{\prime}\right),\label{eq:trace_out_system_orthogonality_condition}\\
\rho_{E_{j}|i,q_{C|i}}^{\left(E_{k}\right)}\rho_{E_{j}|i^{\prime},q_{C|i}^{\prime}}^{\left(E_{k}\right)} & =0,\quad\forall\left(i,q_{C|i}\right)\neq\left(i^{\prime},q_{C|i}^{\prime}\right),\,\forall j\neq k.
\end{align}
\end{proof}

Table~\ref{tab:Summary-of-the-environment-structure} summaries the results in this section.

\section{Objectivity in a dynamic system and two environments \label{sec:Numerical-simulations}}

In the prior sections, we focused primarily on calculating the distinguishability of conditional system and environment states. This distinguishability forms a lower bound to an ideal objective state; however, it is missing a quantification of the non-objective correlations between the system and environments. In this section, we consider a numerical model that allows us to fully explore the divergence from an ideal objective state with the SBS. 

We analyse the broadcast probabilities which show that the information is different in different reference frames. The investigation of mutual information between environments illustrates that strong independence of environments is also not conserved between reference frames. We consider a couple of cases as illustration of phenomena occurring when changing between different reference frames in an information broadcasting scenario. We performed a series of numerical experiments. Note that a computer's memory cannot store an infinite number of data needed to fully describe quantum systems in a continuum of space coordinates.

To provide an conceptual image of the dynamical scenario, we consider a toy-model where the coordinate system is discrete and organized as a ring of size $D$ with all coordinates from the finite set $\{0, \cdots, D - 1\}$ with the metric of the finite field $Z_D$. This coordinate simplification is similar in spirit to the lattice Ising model.

The process of information propagation is governed by relevant Hamiltonians describing the time evolution of interacting subsystems. Here we consider a simple scenario with a central system $S$, interacting with two environments $E_1$ and $E_2$ observed from the point of view of a non-interacting laboratory frame $C$. The reference frame transformation shifts from the point of view of $C$ to the point of view of $E_1$.

The general interaction $H$ between subsystems $S$ and $\{E_i\}_{i=1}^{N}$ can be decomposed into several terms:
\begin{equation}
	\begin{aligned}
		H &= \underbrace{\sum_{i=1}^N \sum_{s=0}^{D-1} \ket{s} \bra{s}_S \otimes H_{E_i}^{(s)}}_\text{central interaction} + \underbrace{\sum_{i=1}^N H_{E_i}}_\text{self-evolution} \\
		& + \underbrace{\sum_{i \neq j = 1}^N H_{E_i, E_j}}_\text{environment interaction} + \underbrace{H_{S, E_1, \cdots E_N}}_\text{global interaction},\label{eq:general_interaction_H}
	\end{aligned}
\end{equation}
where sub-indices enumerates subsystems on which the given part of the total Hamiltonian acts. We note that the form of the central interaction part ensures that the evolution of each of the environments depends on the state of the central system and thus is responsible for imprinting information about it.

In a typical measurement scenario one usually assumes that the evolution is dominated by the central interaction, and then the so-called generalized von Neumann measurement is performed~\cite{Mello2013,Turek2014}. It is reasonable to assume that this part is acting only for a limited period of time, as one expect the measurement to occur after a finite number of time units.

We define the time unit $t=1$ as the time over which the central interaction is active. We also define the energy scale as relative to the strength of the central interaction. We assume that the self-evolution and the interaction of environments is of two orders weaker and the global interaction (that is in most cases a sort of environmental noise) to be weaker of three orders than the central interaction. Since in this paper the Hilbert space is assumed to form a coordinate basis it is natural to pay a particular attention to environment interactions with strengths depending on the distance of subsystems.

To be more specific, the central interaction $H_{S, E_i}^{(s)}$ is defined in a way that after a unit of time the state $\ket{k}_{E_i}$ is transformed to $\ket{k \oplus_D s}_{E_i}$, where $\oplus_D$ is the addition modulo $D$. The environment interaction Hamiltonian $H_{E_i, E_j}$ is defined in a way that propagates jumps of states of a pair of interacting subsystems towards each other with rate of the jumps given by $\frac{0.01}{1 + r}$, where $r$ is the distance between subsystems, and $0.01$ is the coupling constant (two order of magnitudes less than the self-evolution and measurement interaction). A self-evolution of environments allows for jumps towards neighbouring states, leading to a slow spread of the localization.

Since global interaction is conceptualized as being caused by unintended jumps beyond control, the rate of each possible jump is regarded as a uniform random number between $0$ and the coupling constant equal to $0.001$ to model the assumption that this kind of force is of three orders weaker than the measurement interaction.

It has been observed~\cite{Zwolak2009} that the capacity of an environment to receive information about the central system depends on its purity: the higher is the entropy of the subsystem, the less additional information it can gain. In particular one expects that the completely mixed state is not able to perceive the observed entity.

In our investigation we consider various joint states of the central system with two environments. The joint state that maximizes the information flow, and thus is most interesting, is the state:
\begin{equation}
	\label{eq:rho_mpp}
	\rho_\text{mpp} \coloneqq \rho_{\text{mix} S} \otimes \ket{0}\bra{0}_{E_1} \otimes \ket{0}\bra{0}_{E_2},
\end{equation}
where $\rho_\text{mix} \coloneqq \frac{1}{D} \sum_{i=0}^{D-1} \ket{i}\bra{i}$ is the maximally mixed state on the $D$-dimensional ring. To see how mixedness of environments influences information flow we consider a system with slightly blurred environments:
\begin{equation}
	\label{eq:rho_mbb}
	\rho_\text{mbb} \coloneqq \rho_{\text{mix} S} \otimes \rho_{\text{blur} E_1} \otimes \rho_{\text{blur} E_2},
\end{equation}
where
\begin{equation}
	\rho_\text{blur} \coloneqq 0.8 \cdot \ket{0}\bra{0} + 0.1 \cdot \ket{1}\bra{1} + 0.1 \cdot \ket{D - 1}\bra{D - 1}.
\end{equation}
We consider also the cases when only one of the environments is mixed:
\begin{subequations}
	\begin{equation}
		\label{eq:rho_mmp}
		\rho_\text{mmp} \coloneqq \rho_{\text{mix} S} \otimes \rho_{\text{mix} E_1} \otimes \ket{0}\bra{0}_{E_2},
	\end{equation}
	\begin{equation}
		\label{eq:rho_mpm}
		\rho_\text{mpm} \coloneqq \rho_{\text{mix} S} \otimes \ket{0}\bra{0}_{E_1} \otimes \rho_{\text{mix} E_2}.
	\end{equation}
\end{subequations}
Another case that we find interesting to investigate is the situation when the environments are maximally entangled, as this case revealed new phenomena when changing frames in Ref.~\cite{Giacomini2019}. We consider the state:
\begin{equation}
	\label{eq:rho_mEE}
	\rho_\text{mEE} \coloneqq \rho_{\text{mix} S} \otimes \ket{\Phi}\bra{\Phi}_{E_1 E_2},
\end{equation}
where $\ket{\Phi}_{E_1 E_2} \coloneqq \frac{1}{\sqrt{D}} \sum_{i=0}^D \ket{i}_{E_1} \ket{i}_{E_2}$.

We summarize all cases we investigate in the dynamical scenario in the Tab.~\ref{tab:cases}.
\begin{table}[bpt!]
    \caption{Considered dynamical scenarios for a system interacting with two environments $E_1,E_2$. Interaction details for $H_{E_1}$, $H_{E_2}$, and  $H_{E_1, E_2}$ are in main text (following Eq.~\eqref{eq:general_interaction_H}).  The various initial states are given in the main text from Eqs.~\eqref{eq:rho_mpp} to \eqref{eq:rho_mEE}, where the labelling $\rho_{SE_1 E_2}$ denotes where that subsystem is mixed (m), pure (p), blurred/partially mixed (b), or entangled (E).
	\label{tab:cases}}
    \vspace{0.5em}
	\begin{tabular}{|c|c|c|c|c|}
		\hline 
		\begin{tabular}{@{}c@{}}Case \\ label\end{tabular} & self-evolution & \begin{tabular}{@{}c@{}}environment \\ interaction\end{tabular} & \begin{tabular}{@{}c@{}}global \\ interaction\end{tabular} & \begin{tabular}{@{}c@{}}initial \\ state\end{tabular} \\ 
		\hline 
		1.1 & - & - & - & $\rho_\text{mpp}$ \\ 
		\hline 
		1.2 & - & - & - & $\rho_\text{mbb}$ \\ 
		\hline 
		1.3 & - & - & - & $\rho_\text{mEE}$ \\ 
		\hline 
		1.4 & - & - & - & $\rho_\text{mmp}$ \\ 
		\hline 
		1.5 & - & - & - & $\rho_\text{mpm}$ \\ 
		\hline 
		2.1 & $H_{E_1} + H_{E_2}$ & - & - & $\rho_\text{mpp}$ \\ 
		\hline 
		2.2 & random & - & - & $\rho_\text{mpp}$ \\ 
		\hline 
		3.1 & - & $H_{E_1, E_2}$ & - & $\rho_\text{mpp}$ \\ 
		\hline 
		3.2 & - & random & - & $\rho_\text{mpp}$ \\ 
		\hline 
		4 & $H_{E_1} + H_{E_2}$ & $H_{E_1, E_2}$ & random & $\rho_\text{mpp}$ \\ 
		\hline 
	\end{tabular}
	
\end{table}

From the perspective of external observer $C$, the time-dependent tripartite state consists of the central object $S$, and two environments $E_1$ and $E_2$. From frame of the first environment, $E_1$, the relevant state consists of the central object, $S$, the external observer, $C$, and the second environment, $E_2$. We refer to these states as $\rho_{S E_1 E_2}^{(C)}$ and $\rho_{S C E_2}^{(E_1)}$, respectively.

The core part of the SBS is the spectrum of the probability distribution that is broadcast from system to environments. This spectrum is given by
\begin{subequations}
	\begin{equation}
		p_i^{(C)} \coloneqq \bra{i}_S \Tr_{E_1 E_2} \left( \rho_{S E_1 E_2}^{(C)} \right) \ket{i}_S,
	\end{equation}
	\begin{equation}
		p_i^{(E_1)} \coloneqq \bra{i}_S \Tr_{C E_2} \left( \rho_{S C E_2}^{(E_1)} \right) \ket{i}_S,
	\end{equation}
\end{subequations}
where $\Tr_{E_1 E_2}$ and $\Tr_{C E_2}$ denotes partial trace over subsystems $E_1$ and $E_2$, and $C$ and $E_2$, respectively.

The conditional states, c.f. Eq.~\eqref{eq:condStates}, are
\begin{subequations}
	\begin{equation}
		\rho_{E_1|i}^{(C)} \coloneqq \left( 1 / p_i^{(C)} \right) \cdot \bra{i}_S \Tr_{E_2} \left( \rho_{S E_1 E_2}^{(C)} \right) \ket{i}_S,
	\end{equation}
	\begin{equation}
		\rho_{E_2|i}^{(C)} \coloneqq \left( 1 / p_i^{(C)} \right) \cdot \bra{i}_S \Tr_{E_1} \left( \rho_{S E_1 E_2}^{(C)} \right) \ket{i}_S,
	\end{equation}
	\begin{equation}
		\rho_{C|i}^{(E_1)} \coloneqq \left( 1 / p_i^{(E_1)} \right) \cdot \bra{i}_S \Tr_{E_2} \left( \rho_{S C E_2}^{(E_1)} \right) \ket{i}_S,
	\end{equation}
	\begin{equation}
		\rho_{E_2|i}^{(E_1)} \coloneqq \left( 1 / p_i^{(E_1)} \right) \cdot \bra{i}_S \Tr_{C} \left( \rho_{S C E_2}^{(E_1)} \right) \ket{i}_S,
	\end{equation}
\end{subequations}
where $\Tr_{E_1}$, $\Tr_{E_2}$ and $\Tr_{C}$ denotes partial trace over relevant subsystems. For the conditional states $\{\rho_{\cdot|i}^{(\cdot)}\}_{i=0}^{D-1}$ we calculate their two averages, weighted, c.f. Eq.~\eqref{eq:errBound}:
\begin{equation}
	\sum_{i \neq j = 0}^{D - 1} \sqrt{p_i^{(\cdot)} p_{j}^{(\cdot)}}\left\Vert \sqrt{\rho_{\cdot|i}^{(\cdot)}}\sqrt{\rho_{\cdot|j}^{(\cdot)}}\right\Vert _{1},
\end{equation}
and unweighted:
\begin{equation}
	\frac{1}{D(D-1)} \sum_{i \neq j = 0}^{D - 1} \left\Vert \sqrt{\rho_{\cdot|i}^{(\cdot)}}\sqrt{\rho_{\cdot|j}^{(\cdot)}}\right\Vert _{1}.
\end{equation}
For the sake of clarity, cf. Eq.~\eqref{eq:B}, we denote the fidelity terms as:
\begin{subequations}
	\begin{equation}
		B_{E_1}^{(C)}(i,j) \coloneqq \left\Vert \sqrt{\rho_{E_1|i}^{(C)}}\sqrt{\rho_{E_1|j}^{(C)}}\right\Vert _{1},
	\end{equation}
	\begin{equation}
		B_{E_2}^{(C)}(i,j) \coloneqq \left\Vert \sqrt{\rho_{E_2|i}^{(C)}}\sqrt{\rho_{E_2|j}^{(C)}}\right\Vert _{1},
	\end{equation}
	\begin{equation}
		B_{C}^{(E_1)}(i,j) \coloneqq \left\Vert \sqrt{\rho_{C|i}^{(E_1)}}\sqrt{\rho_{C|j}^{(E_1)}}\right\Vert _{1},
	\end{equation}
	\begin{equation}
		B_{E_2}^{(E_1)}(i,j) \coloneqq \left\Vert \sqrt{\rho_{E_2|i}^{(E_1)}}\sqrt{\rho_{E_2|j}^{(E_1)}}\right\Vert _{1}.
	\end{equation}
\end{subequations}

In order to contrast strong versus weak independence between observing subsystems we also calculate the mean  conditional quantum mutual information (Definition~\ref{defn:strongIndependence}):
\begin{subequations}
	\begin{equation}
		I_\text{mean}^{(C)} \coloneqq \sum_{i=0}^{D-1} p_i^{(C)} \left[ H_2\left(\rho_{E_1|i}^{(C)}\right) + H_2\left(\rho_{E_2|i}^{(C)}\right) - H_2\left(\rho_{E_1 E_2|i}^{(C)}\right) \right],
	\end{equation}
	\begin{equation}
		I_\text{mean}^{(E_1)} \coloneqq \sum_{i=0}^{D-1} p_i^{(E_1)} \left[ H_2\left(\rho_{C|i}^{(E_1)}\right) + H_2\left(\rho_{E_2|i}^{(E_1)}\right) - H_2\left(\rho_{C E_2|i}^{(E_1)}\right) \right],
	\end{equation}
\end{subequations}
where $H_2(\cdot)$ is von Neumann entropy and the conditional states are:
\begin{subequations}
	\begin{equation}
		\rho_{E_1 E_2|i}^{(C)} \coloneqq \left( 1 / p_i^{(C)} \right) \cdot \bra{i}_S \rho_{S E_1 E_2}^{(C)} \ket{i}_S,
	\end{equation}
	\begin{equation}
		\rho_{C E_2|i}^{(E_1)} \coloneqq \left( 1 / p_i^{(E_1)} \right) \cdot \bra{i}_S \rho_{S C E_2}^{(E_1)} \ket{i}_S.
	\end{equation}
\end{subequations}
There are dynamical situations when the mean mutual information reaches some value and does not deviate significantly (up to some fluctuations) from it further in time. We refer to this value as the \textit{value of saturation}:
\begin{subequations}
	\begin{equation}
		I_\text{sat}^{(C)} \coloneqq \lim_{T \to \infty} \frac{1}{T} \int_0^T I_\text{mean}^{(C)}(t) dt,
	\end{equation}
	\begin{equation}
		I_\text{sat}^{(E_1)} \coloneqq \lim_{T \to \infty} \frac{1}{T} \int_0^T I_\text{mean}^{(E_1)}(t) dt,
	\end{equation}\label{eq:mutual_info_saturation_E1}
\end{subequations}
from the point of view of $C$ and $E_1$, respectively. The \textit{fluctuations of the mean mutual information} are defined as:
\begin{subequations}
	\begin{equation}
		\sigma_I^{(C)} \coloneqq \left[ \lim_{T \to \infty} \frac{1}{T} \int_0^T \left( I_\text{mean}^{(C)}(t) - I_\text{sat}^{(C)} \right)^2 dt \right]^{\frac{1}{2}},
	\end{equation}
	\begin{equation}
		\sigma_I^{(E_1)} \coloneqq \left[ \lim_{T \to \infty} \frac{1}{T} \int_0^T \left( I_\text{mean}^{(E_1)}(t) - I_\text{sat}^{(E_1)} \right)^2 dt \right]^{\frac{1}{2}},
	\end{equation}\label{eq:mutual_info_fluctuations}
\end{subequations}
for reference frames of $C$ and $E_1$, respectively. We also define the \textit{time of saturation} $t_\text{sat}^{(C)}$ from the perspective $C$ ($t_\text{sat}^{(E_1)}$ from the perspective $E_1$), as a time when the mean mutual information reaches the value $I_\text{sat}^{(C)} - \sigma_I^{(C)}$ ($I_\text{sat}^{(E_1)} - \sigma_I^{(E_1)}$) for the first time.

We performed a series of numerical simulations in order to investigate how well the SBS form is preserved in both the frames of the external observer, $C$, and the first environment, $E_1$. For the cases discussed below we take $D = 12$. This dimension has been chosen as a compromise between computational effort and modelling the dependence of behaviours of subsystems on their spatial separations. The long time averages are calculated over time points between $50000$ and $1000000$ with time step $50000$.

The reference frame transformation $\hat{S}_{\text{position}}^{\left(C\rightarrow E_1\right)}$ satisfying Eq.~\eqref{eq:transform} is, in this case, a permutation of indices of rows and columns of a tripartite $D^3$ dimensional density matrix. As such the spectrum of the initial density matrix remains invariant. For the case of dimension $D=12$ we have found that the unitary transformation over $D^3 = 1728$ dimensional joint tri-coordinate space has character (trace) $144$, contains $144$ irreducible subspaces of dimension $1$, and $792$ irreducible subspaces of dimension $2$. We have directly checked that $\chi \left[ \hat{S}_{\text{position}}^{\left(C\rightarrow E_1\right)} \right] = D^2$, where $\chi[\cdot]$ is the character (trace) of a transformation, for $D \leq 25$ and conclude that this is a general property of reference frame transformations as~$D \to \infty$.

Below we summarize how the properties relevant for perceived objectivity behave in our dynamical model. The change in the broadcast spectrum of probabilities (Sec.~\ref{subsec:probabilities}), the dynamics and overall volume of mutual information (Sec.~\ref{subsec:mutual_info}), the varying distinguishability of subsystems (Sec.~\ref{subsec:Distinguishability}), and the Holevo and quantum mutual information between central system and subenvironments (Sec.~\ref{subsec:Holevo_and_mutual_info}) in different frames provides a concrete illustration of the main premise of this paper: what is objective from one point of view may not be \textit{objective} from another point of view. All the cases referred to in the following section have been labelled in Tab.~\ref{tab:cases}. The full details are described in  Appendix~\ref{app:dynamic}.

\subsection{Probabilities\label{subsec:probabilities}}

\begin{figure}
	\centering
	\includegraphics[width=0.95\linewidth]{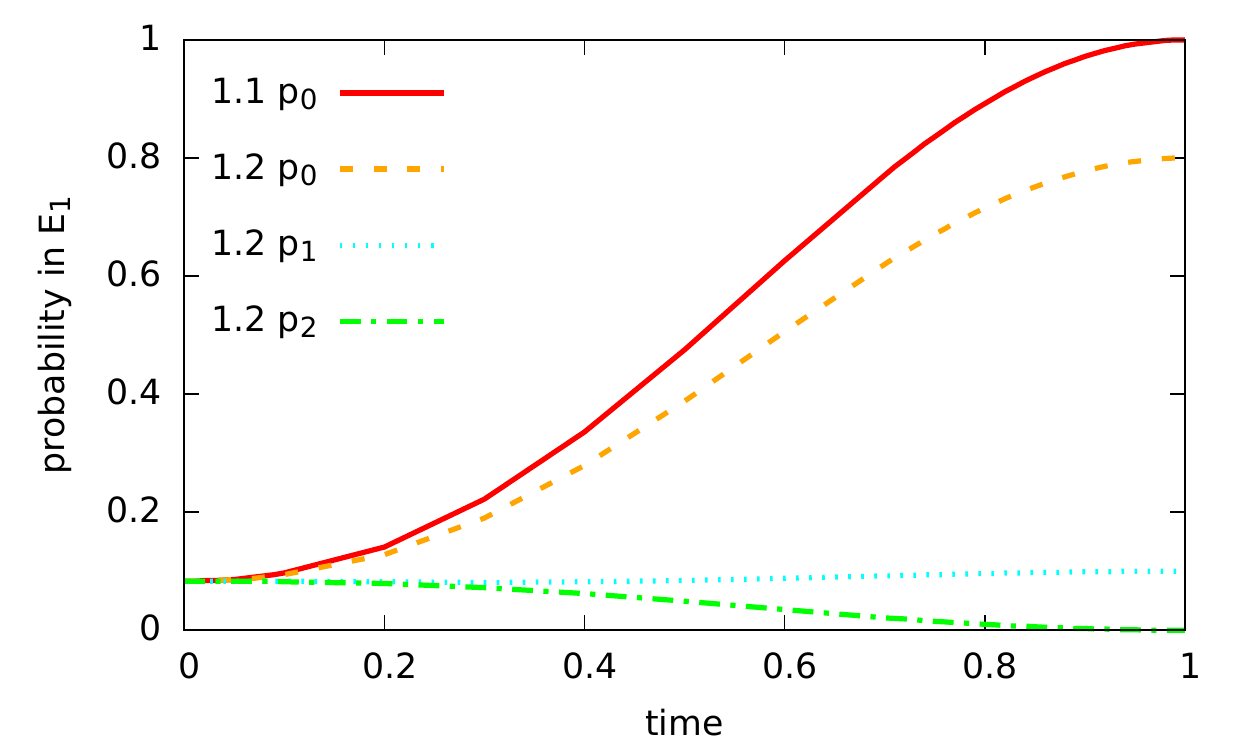}
	\caption{(Color on-line) Selected probabilities in the reference frame $E_1$ for Cases 1.1 and 1.2. Note that for Cases 1.3 and 1.4 the probability $p_0^{(E_1)} = \frac{1}{D} \approx 0.083$ and is constant in time. For the Case 1.5 $p_0^{(E_1)}$ is identical as in the Case 1.1. (Cases given in Tab.~\ref{tab:cases})}
	\label{fig:plotprobpie1t1}
\end{figure}

\begin{figure}
	\centering
	\includegraphics[width=0.95\linewidth]{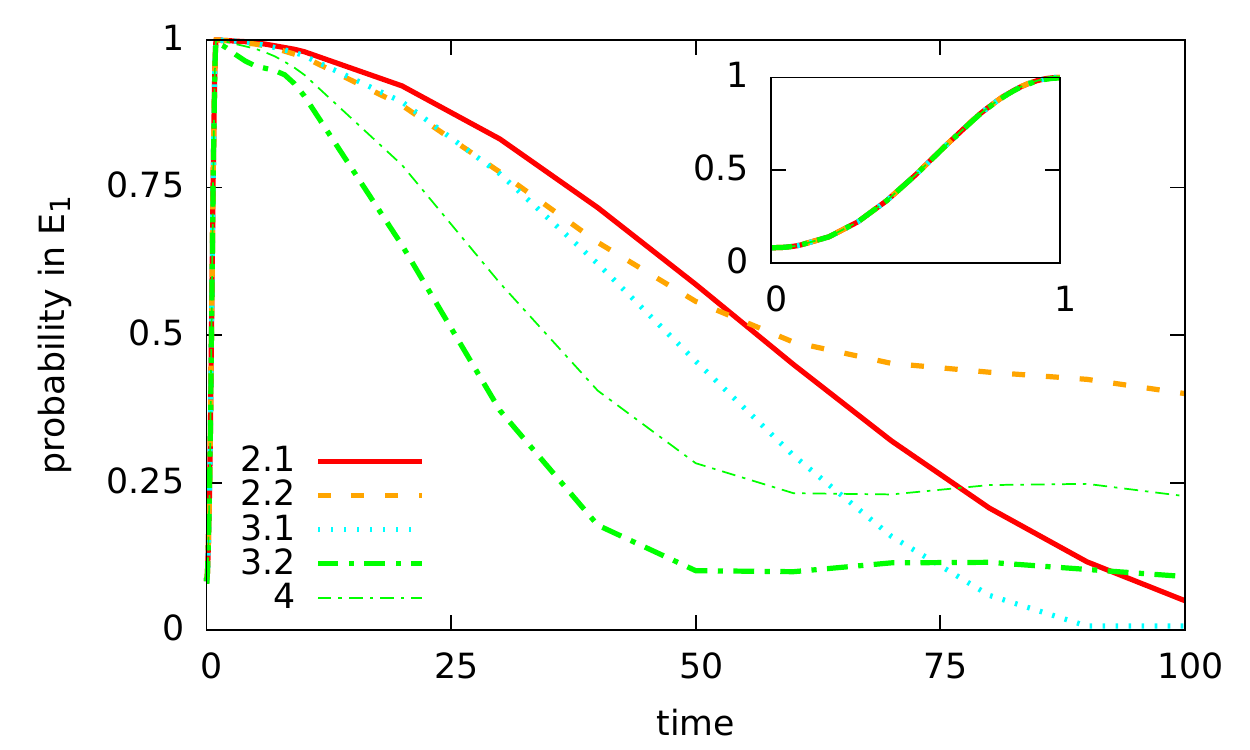}
	\caption{(Color on-line) Probabilities $p_0^{(E_1)}$ for cases from the groups 2, 3, and 4 in the reference frame of $E_1$, which have various interactions in addition to central measurement Hamiltonian. Inset: probabilities during time $[0,1]$ where the measurement-like Hamiltonian dominates.  (Cases given in Tab.~\ref{tab:cases})}
	\label{fig:plotprobp1e1tm}
\end{figure}

The probabilities of the central system are key for the spectrum in an objective state. From our numerical calculations, we find that all probabilities (of the central system spectrum) $\{p_i^{(C)}\}_{i=0}^{D-1}$ from laboratory $C$'s point of view are constant and uniform over time, except in Case 4 (cases given in Tab. \ref{tab:cases}) where the added global interaction influences the central subsystem and thus modifies its spectrum.

The system-environment generalised measurement interaction has been designed such that the effective measurement occurs by time~$t=1$. This occurs without disturbance in Cases 1.1 and 1.5 (Tab.~\ref{tab:cases}), where there are no other interactions, and where the initial state of $E_1$ in $C$'s frame of reference, $\rho_{E_1}^{(C)}(t) \coloneqq \Tr_{S E_2} \left( \rho_{S E_1 E_2}^{(C)}(t) \right)$, is pure (see Fig.~\ref{fig:plotprobpie1t1}). In these two cases, $p_0^{(E_1)} \approx 1$  at time~$t=1$, \emph{i.e.} implying that the central system state is close to pure at the end of the measurement, which corresponds to a trivial kind of objectivity in the frame of $E_1$ (much like Example \ref{eg:degenerate_localised}). Alternatively, we can say that $E_1$'s capacity was not \textit{used}.

If the initial environment state $\rho_{E_1}^{(C)}$ is slightly mixed in Case 1.2, the value of the probability $p_0^{(E_1)}$ in the frame of $E_1$ diverges from $1$, proportional to the mixedness in the original $E_1$ state. In Cases 1.3 and 1.4, when $p_0^{(E_1)}$ was either maximally mixed or a part of a maximally entangled state (which means that the local state of $E_1$ is identical) the values of probabilities $\{p_i^{(E_1)}\}_{i=0}^{D-1}$ remain uniform at the time~$t=1$, meaning there was no information transfer at all.

The probabilities in the frame of $E_1$ for all other cases behave very similarly to the Case 1.1 up to time~$t=1$, despite their interactions, since their evolution on the short time scale is still dominated by the measurement interaction (see Fig.~\ref{fig:plotprobp1e1tm}).

\subsection{Strong independence and the conditional mutual information\label{subsec:mutual_info}}

The mean mutual information between the conditional environment informs us on whether the subsystems have strong or weak independence, the former of which is a condition of spectrum broadcast structure: a small (ideally zero) mean mutual information denotes strong independence.

In Figs~\ref{fig:it1} and~\ref{fig:ie1tm}, we give the plots for how the mean mutual information $I_\text{mean}^{(E_1)}$ behaves over time, in the frame of $E_1$. We find that the mean mutual information typically starts from the value $I_\text{mean}=0$, reaches local maximum close to $1.7$ exactly at the time $t=0.5$ and returns close to $0$ for the time~$t=1$. The only exception from this behaviour is the Case 1.4 (where the initial state of $E_1$ was mixed in the lab frame $C$), and we are observing the system $\rho_{S C E_2}^{(E_1)}$ from a point of view of a completely random observer. In that case the mutual information starts with the maximal value equal $3.585$ and gradually drops. For all cases, $I_\text{mean}^{(E_1)}(t=1) \approx 0$ (with very small non-zero value in the cases with random interactions). 

\begin{figure}
	\centering
	\includegraphics[width=0.95\linewidth]{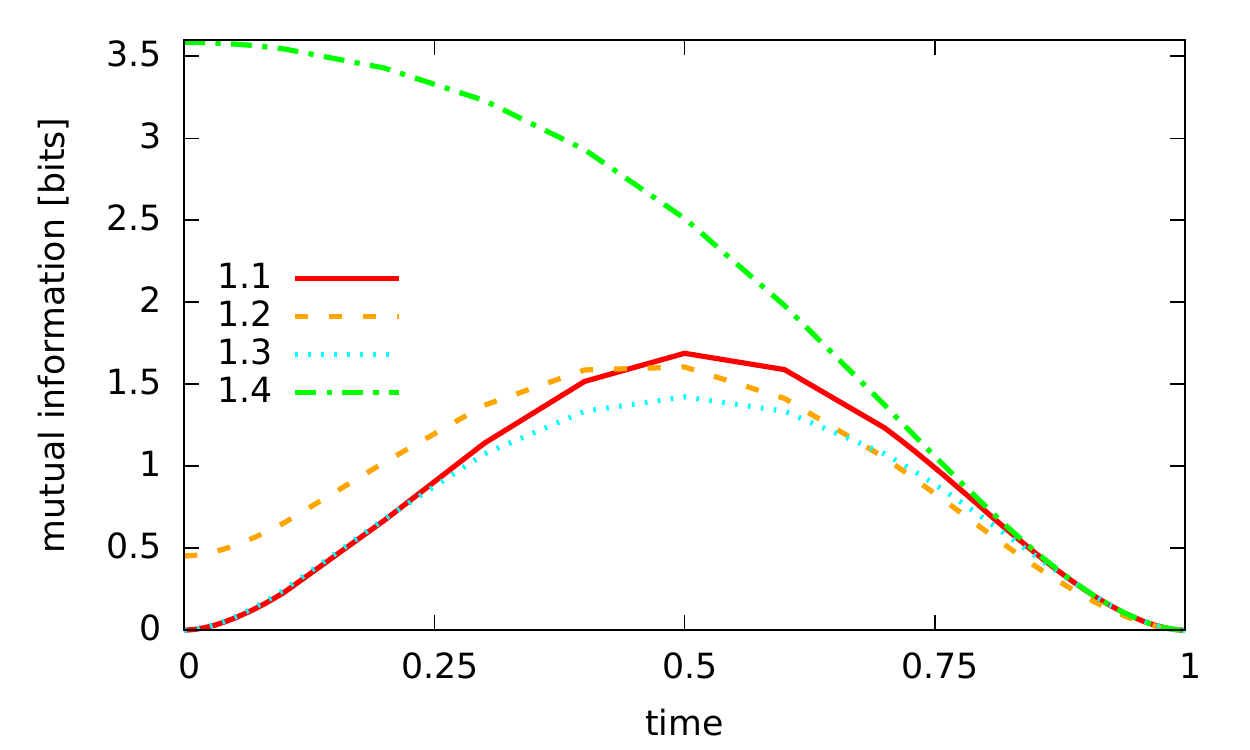}
	\caption{(Color on-line) Plot of $I_\text{mean}^{(E_1)}$ (the mean mutual information seen from the reference frame of the first environment, $E_1$) for Cases 1.1, 1.2, 1.3 and 1.4. In the Case 1.5 the function is constant and equal to $0$. This shows the cases when $I_\text{mean}^{(E_1)}(t=1) = 0$. The plot for the Case 1.1 is identical for Cases 2.1 and 3.1. (Cases given in Tab.~\ref{tab:cases}). When the mutual information is low, the environments have strong independence.}
	\label{fig:it1}
\end{figure}

\begin{figure}
	\centering
	\includegraphics[width=0.95\linewidth]{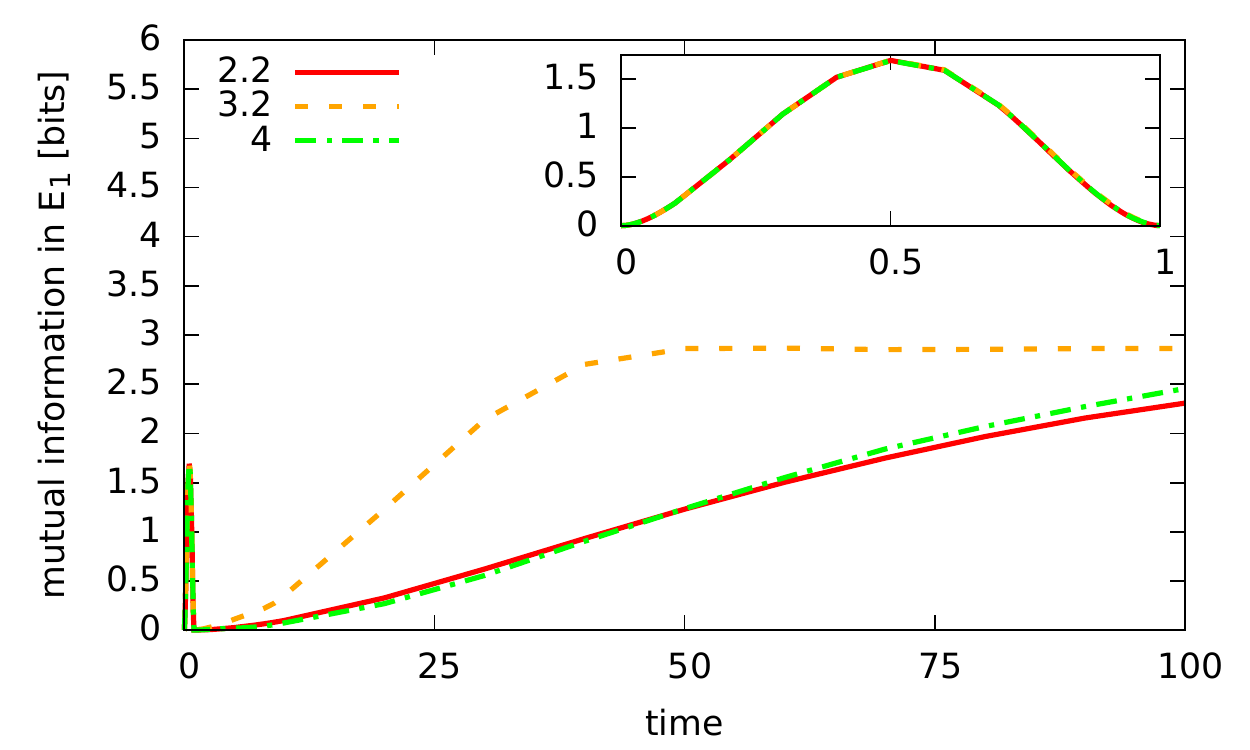}
	\caption{(Color on-line) Plot of mean mutual information in the reference frame of $E_1$, $I_\text{mean}^{(E_1)}$ for Cases 2.2, 3.2, and 4, cf. Tab.~\ref{tab:MI_E1}.  Inset: mean mutual information during times $t=[0,1]$ where the measurement-like Hamiltonian dominates. (Cases given in Tab.~\ref{tab:cases}). As the mutual information is large as time increases, the environments do not have strong independence.}
	\label{fig:ie1tm}
\end{figure}

The Cases 1.1, 2.1, and 3.1, which all have the same initial states, have very similar values of $I_\text{mean}^{(E_1)}$ in each moment of time, which can be seen in Fig.~\ref{fig:it1}. Furthermore, the mutual information is $0$ after time~$t=1$ in reference frame $E_1$. Thus, while the observing subsystems $C$, $E_2$ initially develop conditional correlations during times $t=[0,1]$, they satisfy strong independence thereafter, a necessary condition for ideal objectivity (Definition~\ref{defn:strongIndependence}).

The second common pattern is shown in Fig.~\ref{fig:ie1tm}, where mean mutual information in the reference frame of $E_1$, after approaching the value $0$ at the time~$t=1$, will gradually increase to some saturation level, with slight fluctuations. This happens in Cases 2.2, 3.2, and 4, where some random Hamiltonian is present (cf.~Tab.~\ref{tab:MI_E1}). In all other cases the mean mutual information in reference frame $E_1$ is $0$ after the time~$t=1$.

\begin{table}[bpt!]
\caption{Dynamics of mean mutual information $I_\text{sat}^{(E_1)}$ in the reference of frame of the first environment, $E_1$; $\sigma_I^{(E_1)}$ describes the fluctuation and $t_\text{sat}^{(E_1)}$ is the time of saturation (cf.~Fig.~\ref{fig:ie1tm} and Eqs.~\eqref{eq:mutual_info_saturation_E1},~\eqref{eq:mutual_info_fluctuations}).}\label{tab:MI_E1}
	\begin{tabular}{|c|c|c|c|c|c|c|}
		\hline 
		& 2.2 & 3.2  & 4 \\ \hline
		$I_\text{sat}^{(E_1)}$ & $2.513$ & $2.867$ & $3.508$ \\ \hline
		$\sigma_I^{(E_1)}$ & $0.168$ & $0.022$ & $0.004$  \\ \hline
		$t_\text{sat}^{(E_1)}$ & $150$ & $50$ & $500$ \\ \hline
	\end{tabular}
	
\end{table}

On the other hand, from the point of view of the external laboratory $C$, the mean mutual information is constant when there is no interaction between environments (Case groups 1 and 2). In this reference frame, cases with inter-environmental interaction (groups 3 and 4) lead to  the mean mutual information gradually increasing (without local maxima) till the point of saturation, $t_\text{sat}^{(C)}$. The parameters of this behaviour are given in Tab.~\ref{tab:MI_C}.

We see that the tripartite states typically have strong independence while in the lab frame $C$, but this weakens when moving to the environment frame $E_1$.

\begin{table}[bpt!]
\caption{Dynamics of mean mutual information in the frame of reference of the external laboratory $C$. Values at the times $t=0.5,1$ shows the gradual increase (cf.~Fig.~\ref{fig:ictm} in the Appendix~\ref{app:dynamic}).}\label{tab:MI_C}
	{\small
	\begin{tabular}{|c|c|c|c|c|c|c|c|c|}
		\hline 
		& 3.1 & 3.2 & 4 \\ \hline 
		$I_\text{mean}^{(C)}(0.5)$ & $0.002$ & $0.023$ & $0.004$ \\ \hline 
		$I_\text{mean}^{(C)}(1)$ & $0.006$ & $0.070$ & $0.012$ \\ \hline 
		$I_\text{sat}^{(C)}$ & $2.769$ & $5.740$ & $3.521$  \\ \hline 
		$\sigma_I^{(C)}$ & $0.949$ & $0.032$ & $0.004$ \\ \hline 
		$t_\text{sat}^{(C)}$ & $30$ & $50$ & $60$ \\ \hline 
	\end{tabular}
	}
\end{table}

We observe that in the short time range, before the interaction of environments with central system has fully occurred, environments are more independent in the laboratory's frame, whereas after the interaction the independence is stronger for environmental frame.

For long time scale we distinguish two situations: If the environmental interaction decreases with distance (Case 3.1) then the strong independence between $C$ and $E_2$ occurs in the environment frame $E_1$, but not between $E_1$ and $E_2$ in frame $C$. Meanwhile, if the environment interaction is random and does not decrease with distance (Case 3.2), there is no strong independence in either frames. In this second situation, the mutual information happens to be almost exactly twice as large (differing by $2\%$ at most) in $C$'s frame compared to $E_1$'s frame, suggesting that there is greater independence in frame $E_1$. When global interaction is present, then the mutual information is the same in both frames (differing by $0.1\%$ at most) for all times.

Overall, in the cases we consider, strong independence can be maintained in original frame $C$, but not the environment frame $E_1$. Furthermore, environment-environment interactions destroy strong independence.

\subsection{Distinguishability\label{subsec:Distinguishability}}

Distinguishability is crucial for observers to determine the central system's spectrum. In the reference frame of $C$ the upper bound that describes the error in distinguishing conditional states, Eq.~\eqref{eq:errBound}, is equal $0$ at the time~$t=1$ for Cases 1.1, 2.1, and close to $0$ for Cases 2.2 and 3.1. It is also $0$ for the subsystem $E_2$ in Case 1.4 (when $E_2$ is pure and $E_1$ is maximally mixed) and for subsystem $E_1$ in Case 1.5 (when $E_1$ is pure and $E_2$ is maximally mixed). That is, the distinguishability error is (close to zero) for situations without random inter-environmental interactions.

For the other cases which \emph{do} have random interactions, error the bound is significantly higher, \textit{viz.} for the Case 3.2 the bound is between $0.38$ and $0.39$, equal to $0.169$ for the Case 4.

\begin{table}[bpt!]
\caption{Upper bound to distinguishability error, Eq.~\eqref{eq:errBound}, at the time~$t=1$ for cases without random inter-environmental interactions, in the reference frame of $C$ for subsystems $E_1$ and $E_2$ and in the reference frame of $E_1$ for subsystems $C$ and $E_2$.}\label{tab:errC} \label{tab:errE1}
	\begin{tabular}{|c|c|c|c|c|c|c|}
		\hline 
		frame & subsystem & 2.1& 2.2 & 3.1  \\ \hline 
		C &$E_1^{(C)}$ & $0.000$ & $0.28$ & $0.050$  \\ \hline 
		C &$E_2^{(C)}$ & $0.000$ & $0.025$ & $0.050$  \\ \hline 
		$E_1$ & $C^{(E_1)}$ & $0.040$ & $0.109$ & $0.051$ \\ \hline 
		$E_1$ & $E_2^{(E_1)}$ & $0.000$ & $0.000$ & $0.011$  \\ \hline 
	\end{tabular}
	
\end{table}

Similarly, from the perspective of $E_1$, the distinguishability error upper bound Eq.~\eqref{eq:errBound} is low (below $0.05$ for at least one subsystem $C$ or $E_2$) when there are no random interactions and the initial state of $E_1$ is pure: these correspond to Cases 1.1, 2.1, 2.2, and 3.1, and also 1.5 (that was $0$ for $E_1$ and above $1$ for the initially maximally mixed $E_2$). This is summarised in Tab.~\ref{tab:errE1}. 

Overall, the effect of moving frame is thus: while it is easier to distinguish $E_1$ in lab $C$ frame, while $E_2$ has the lower distinguishability error when in the $E_1$ frame.

\subsection{Holevo information and mutual information between central system and environments\label{subsec:Holevo_and_mutual_info}}

The analysis in the previous subsections have shown that the SBS structure is generally not preserved between reference frames. Another aspect of quantum Darwinism and SBS is the information propagation measured by the  quantum mutual information, and by the (classical) Holevo information transfer from the central system to environments. To this end, we calculate the quantum mutual information, as well as the Holevo information by considering a hypothetical measurement of the central system in its eigenbasis, and the ensemble of steered states of each of the environments.

Fig.~\ref{fig:QMI_S_E1} shows  the quantum mutual information between the central system $S$ and subsystem $E_1$ ($C$) in the reference frame of $C$ ($E_1$). Both Holevo information (figure given in Appendix~\ref{app:dynamic} as it is extremely similar) and quantum mutual information of both sub-environments are equal in frame $C$ and equal for the environment $E_2$ in frame $E_1$ if there is no global interaction (Cases 1, 2, 3), thus satisfying one requirement for SBS. The low quantum mutual information in frame $E_1$ compared to high quantum mutual information in frame $C$ also suggests that the central system and $C$ are trivially objective, i.e. that only one probability in the spectrum is substantial (cf. Sec.~\ref{subsec:probabilities}).

On the other hand, the Holevo and quantum mutual information differ for the subsystem $C$ in frame $E_1$ in most Cases. For the Case 1.3 (with an initially entangled environment state) the quantum mutual information is constant and is equal to the half of the mutual information of a maximally entangled state, and Holevo information is $0$. If there is self-evolution or environment interaction (Cases 2 and 3), the quantum mutual information is also significantly higher that Holevo information, with exception for the Case 2.1 when they are equal. In general, the state $C$ in frame $E_1$ develops quantum correlations with the central system, reflecting the results from the prior sections.

\begin{figure}
	\centering
	\includegraphics[width=0.99\linewidth]{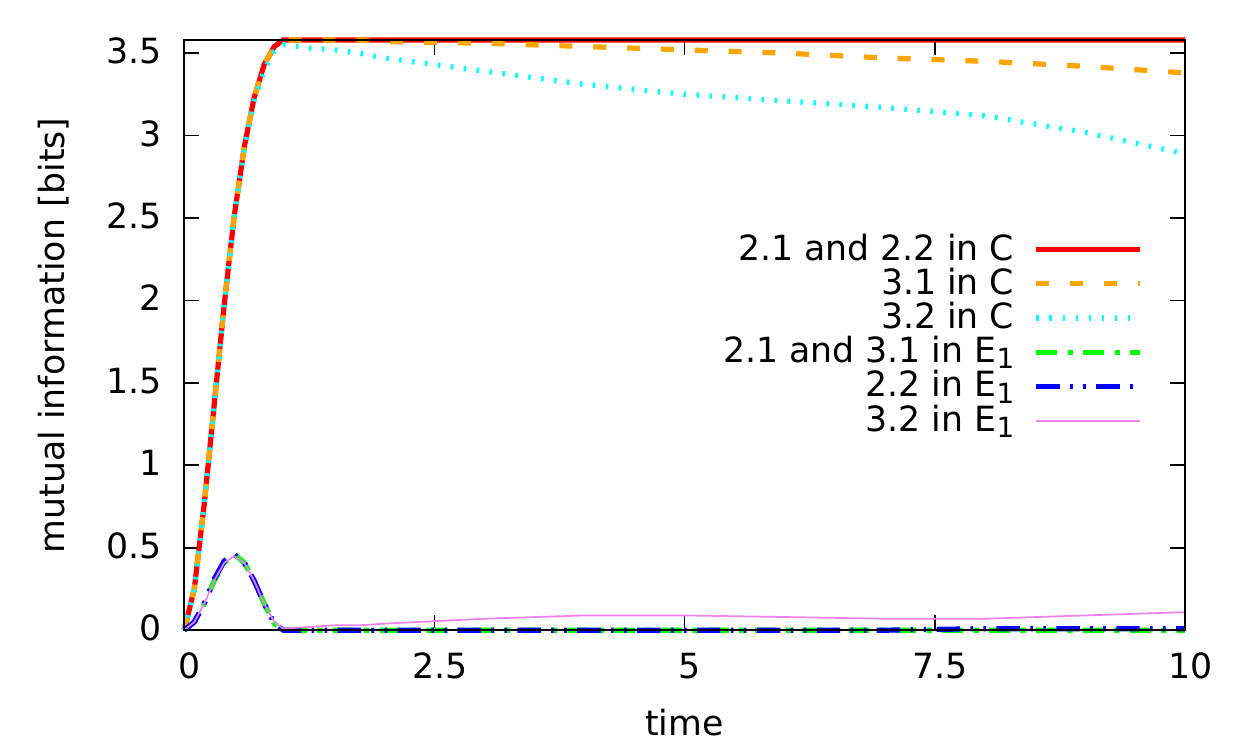}
	\caption{(Color on-line) The quantum mutual information between the central system $S$ and subsystem $E_1$ ($C$) in the reference frame of $C$ ($E_1$). A high quantum mutual information often suggests objectivity. (Cases given in Tab.~\ref{tab:cases})}
	\label{fig:QMI_S_E1}
\end{figure}

\section{Conclusion\label{sec:Conclusion}}

We examined the transformation of objectivity in different quantum reference frames. We used the quantum reference frame formalism of \citet{Giacomini2019}, in which local and global properties are frame-dependent and can interchange, and analysed the structure of objective states in the frame of their environments.

Under perfect localisation, we showed how non-degenerate relative positions is the key factor in ensuring that objectivity is consistent across different quantum reference frames. This can be done by randomly choosing all required positions from a continuous interval\textemdash as the probability of two random real numbers being equal is zero. Environment-state coherences introduce entanglement between system and environments, which then require discarding of environments to remove. Meanwhile, environment-state statistical mixedness, \emph{i.e.} internal classical noise, introduces new classical correlations in addition to the original classical objective correlations. Thus, in general, each reference frame has a different set of objective information that involves the original statistics of the environment associated with the reference frame. Nevertheless, the original system information is recoverable by taking the appropriate marginal of the information distribution. Hence, the objective information of the system in the original laboratory frame is unique and exists consistently in all frames---although typically the system-environment state no longer has an objective state form in other frames.

We then considered systems and environments with a non-zero spread over position. Distinguishability requires either sufficiently separated (non-degenerate) relative positions, or large differences in spread of different conditional states. Objectivity becomes ``blurred'' in different frames due to the continuous spread in the environment states, and the greater the blurring, the large the macrofractions are required in the new frame in order to recover a form of objectivity.

We found that the best candidate of objective information finds the system in conditional mixed states that are distinguishable, rather than the conditional pure system states strictly required of quantum Darwinism. This suggests a generalised objectivity, \emph{i.e.} where the system is conditionally mixed and perfectly distinguishable like the environments. Furthermore, only a strict subset of all objective states are perfect robust across quantum reference frames\textemdash when its observed environments are also objective, that is having an invariant SBS \citep{Le2019b}. This suggests that objective states is consistent under different frames only if it system and environments have similar structure, \emph{i.e.} all conditionally pure, or all conditionally mixed.

We then proved that perfect objectivity holds in all lab and environment frames if and only if the environment states are pure and perfectly localised in the position basis conditioned on $i$ of the system information. An open question is how states close to this specialised form of SBS behave under quantum reference frame transformations.

Finally, we examined the dynamical emergence of objectivity in different quantum reference frames, explicitly demonstrating how changing frames affects the objective probabilities and degrades spectrum broadcast structure by weakening strong independence and reducing distinguishability.

Quantum Darwinism is an approach towards understanding the quantum-to-classical transition, examining how hypothetical observers may acquire objective information about a common system. Observers may have different reference frames, which subsequently affects the information they can obtain. We have shown that the system's objective information is recoverable in all frames, despite the interchangeability of the local and global statistics. At the same time, our work demonstrates the rise of extra frame-dependent objective information, and the robustness (or otherwise) of various objective state structures. In the strictest sense, objectivity is subjective across quantum reference frames. This work opens up the pathway to understanding quantum Darwinism and its intersection with relativity. On a more philosophical level the considerations are tightly related to the problem of Wigner's friend~\cite{Brukner2015} gedanken experiment.

After the completion of this manuscript, we were made aware of independent work on a similar topic: Ref.~\cite{Tuziemski2020} examines dynamical aspects of quantum reference frame transformations and objectivity, and thus their results are complementary to our results presented here.

\section*{Acknowledgements}

The support by the Foundation for Polish Science through IRAP project co-financed by the EU within the Smart Growth Operational Programme (contract no. 2018/MAB/5) is acknowledged. This work was also supported by the Engineering and Physical Sciences Research Council {[}grant number EP/L015242/1{]}.

\appendix

\section{Objectivity with continuous probabilities \label{app:Objectivity-with-continuous}}

In the main paper, the branch structure of the objective state is typically discrete, indexed by $\left\{ i\right\} $ and summed. However, we can consider the more general situation in which this $\left\{ i\right\} $ becomes continuous, even in the laboratory frame. A GHZ-like continuous objective state is the following:
\begin{align}
\rho_{SE_1E_2\cdots E_N}^{(C)}= & \int dx_S\psi\left(x_S\right)\ket{x_S}\bra{x_S}_S\nonumber \\
 & \otimes\ket{\phi_{1}\left(x_S\right)}\bra{\phi_{1}\left(x_S\right)}_{E_1}\nonumber \\
 & \otimes\cdots\otimes\ket{\phi_{N}\left(x_S\right)}\bra{\phi_{N}\left(x_S\right)}_{E_N},\label{eq:GHZ_continuous_probabilities}
\end{align}
where $\psi\left(x_S\right)$ are the continuous probabilities, and $\left\{ \phi_i\right\} _i$ are bijective (one-to-one and onto) functions. One-to-oneness is necessary so that the environments $E_i$ positions are unique for any different system position $x_S$, and hence are always perfectly correlated with the system position $x_S$. For simplicity, we will also take $\left\{ \phi_i\right\}_i$ being onto, which ensures that the inverse is also one-to-one.

In the quantum reference frame of the environment $E_1$, the joint state has form:
\begin{align}
 & \rho_{SC E_2\cdots E_N}^{\left(E_1\right)}\nonumber \\
 & =\int dq_C\psi\left(\phi_{1}^{-1}\left(-q_C\right)\right)\nonumber \\
 & \phantom{=}\times\ket{\phi_{1}^{-1}\left(-q_C\right)+q_C}\bra{\phi_{1}^{-1}\left(-q_C\right)+q_C}_S\otimes\ket{q_C}\bra{q_C}_C\nonumber \\
 & \phantom{=}\otimes\bigotimes_{j=2}^{N}\ket{\phi_{j}\left(\phi_{1}^{-1}\left(-q_C\right)\right)+q_C}\bra{\phi_{j}\left(\phi_{1}^{-1}\left(-q_C\right)\right)+q_C}_{E_j}.
\end{align}
The original probabilities $\psi\left(\cdot\right)$ still exist---as the original conditional environment states are pure in the position basis (c.f. Thm.~\ref{thm:perfect_objectivity_discrete}). In the frame $E_1$, there is an integral over $q_C$ rather than $x_S$. In this particular scenario, objectivity requires distinguishability for different $q_C$. Hence $\phi_{1}^{-1}\left(-q_C\right)+q_C$ and $\left\{ \phi_{j}\left(\phi_{1}^{-1}\left(-q_C\right)\right)+q_C\right\} _{j}$ must also be one-to-one functions\textemdash so that none of the new positions become degenerate (and subsequently reduce conditional distinguishability).

Note that the composition of one-to-one functions is one-to-one; but sum of one-to-one functions is not necessarily one-to-one: a sufficient but not necessary condition is that $\dfrac{d}{dx}\phi_i\left(\phi_{1}^{-1}\left(x\right)\right)>1$ or that $\dfrac{d}{dx}\phi_i\left(\phi_{1}^{-1}\left(x\right)\right)<1$ for all $x$. Hence the state given in Eq.~(\ref{eq:GHZ_continuous_probabilities}) is \emph{not} always objective for any set of one-to-one functions $\psi(\cdot)$, $\phi_j(\cdot)$.

\section{Proof of conditions for perfect objectivity\label{app:perfect_proofs}}

\subsection{Proof of Theorem \ref{thm:perfect_objectivity_discrete}\label{app:proof_theorem}}

Here, we prove the state structure and conditions for an SBS state to be objective, with the same objective information, in all laboratory and environment quantum reference frames.

\begin{proof}
In the new environment frame $E_{1}$, the system-environment state has general form:
\begin{align}
 & \rho_{SCE_{2}\cdots E_{N}}^{\left(E_{1}\right)}\nonumber \\
 & =\sum_{i}p_{i}\sum_{q_{C},q_{C}^{\prime}}\sum_{q_{S},q_{S}^{\prime}}\psi\left(q_{S}-q_{C}|i\right)\psi^{*}\left(q_{S}^{\prime}-q_{C}^{\prime}|i\right) \ket{q_{S}} \bra{q_{S}^{\prime}}_{S} \nonumber \\
 & \quad\otimes t\left(-q_{C},-q_{C}^{\prime}|i,j=1\right)\ket{q_{C}}\bra{q_{C}^{\prime}}_{C}\nonumber \\
 & \quad\otimes\bigotimes_{j=2}^{N}\sum_{q_{E_{j}},q_{E_{j}}^{\prime}}t\left(q_{E_{j}}-q_{C},q_{E_{j}}^{\prime}-q_{C}^{\prime}\Big|i,j\right)\ket{q_{E_{j}}}\bra{q_{E_{j}}^{\prime}}_{E_{j}},\label{eq:general_system_state_E1_frame}
\end{align}
where note that in $t(-q_C,-q_C^\prime |i,j=1)$, only $j=1$ is fixed. The general reduced system state is:
\begin{align}
\rho_{S}^{\left(E_{1}\right)} & =\sum_{i}p_{i}\sum_{q_{C}}t\left(-q_{C},-q_{C}|i,j=1\right)\nonumber \\
 & \phantom{=}\times\sum_{q_{S},q_{S}^{\prime}}\psi\left(q_{S}-q_{C}|i\right)\psi^{*}\left(q_{S}^{\prime}-q_{C}|i\right)\ket{q_{S}}\bra{q_{S}^{\prime}}_{S}.\label{eq:general_reduced_system_state_discrete}
\end{align}
If the \emph{same }initial objective information $\left\{ p_{i}\right\} _{i}$ is maintained, the system must have $\left\{ p_{i}\right\} _{i}$ as its spectrum, even in the frame $E_{1}$. Hence, we must be able to decompose it as $\rho_{S}^{\left(E_{1}\right)}\overset{!}{=}\sum_{i}p_{i}\ket{\tilde{\psi}_{i}^{S}}\bra{\tilde{\psi}_{i}^{S}}_{S}$ , where $\ket{\tilde{\psi}_{i}^{S}}=\sum_{q_{S}}T\left(q_{S}\right)\ket{q_{S}}_{S}$ are the new eigenstates of $S$ in the frame $E_{1}$ with some coefficients
$T\left(q_{S}\right)$:
\begin{align}
\rho_{S|i}^{\left(E_{1}\right)} & =\sum_{q_{C}}t\left(-q_{C},-q_{C}|i,j\right)\nonumber \\
 & \phantom{=}\times\sum_{q_{S},q_{S}^{\prime}}\psi\left(q_{S}-q_{C}|i\right)\psi^{*}\left(q_{S}^{\prime}-q_{C}|i\right)\ket{q_{S}}\bra{q_{S}^{\prime}}_{S}\\
 & =\left(\sum_{q_{S}}\ket{q_{S}}_{S}\right)\left(\sum_{q_{S}^{\prime}}\bra{q_{S}^{\prime}}_{S}\right)\nonumber \\
 & \qquad\times\sum_{q_{C}}\left(\begin{array}{r}
\psi\left(q_{S}-q_{C}|i\right)\psi^{*}\left(q_{S}^{\prime}-q_{C}|i\right)\\
\times t\left(-q_{C},-q_{C}|i,j=1\right)
\end{array}\right)\label{eq:temporary_eq}\\
 & \overset{!}{=}\sum_{i}p_{i}\ket{\tilde{\psi}_{i}^{S}}\bra{\tilde{\psi}_{i}^{S}}_{S}\\
 & =\left(\sum_{q_{S}}T\left(q_{S}\right)\ket{q_{S}}_{S}\right)\left(\sum_{q_{S}^{\prime}}T^{*}\left(q_{S}^{\prime}\right)\bra{q_{S}^{\prime}}_{S}\right).
\end{align}
The coefficient $\sum_{q_{C}}(\cdots)$ term in Eq.~(\ref{eq:temporary_eq}) must be factorisable into \emph{independent} terms involving $q_{S}$ and $q_{S}^{\prime}$:
\begin{align}
T\left(q_{S}\right)T^{*}\left(q_{S}^{\prime}\right) & \overset{!}{=}\sum_{q_{C}}\psi\left(q_{S}-q_{C}|i\right)\psi^{*}\left(q_{S}^{\prime}-q_{C}|i\right)\nonumber \\
 & \qquad\qquad\times t\left(-q_{C},-q_{C}|i,j=1\right).
\end{align}
The only way to factorise this is if there is only \emph{one} nonzero term in the sum, i.e. if and only if
\begin{equation}
t\left(-q_{C},-q_{C}|i,j=1\right)\overset{!}{=}\delta\left(-q_{C}-x_{E_{j=1}|i}\right).
\end{equation}
Thus, the original conditional environment states $\rho_{E_{1}|i}$ of environment $E_1$ are pure \emph{and} incoherent in the $x$ basis: 
\begin{align}
\rho_{E_{1}|i}^{(C)} & =\sum_{x_{E_{1}},x_{E_{1}}^{\prime}}t\left(x_{E_{1}},x_{E_{1}}^{\prime}|i,j=1\right)\ket{x_{E_{1}}}\bra{x_{E_{1}}^{\prime}}_{E_{1}}\nonumber \\
 & \overset{!}{=}\ket{x_{E_{1}|i}}\bra{x_{E_{1}|i}}.
\end{align}
This holds analogously for objectivity in any other environment quantum reference frame. Hence, we require that all conditional environment states are pure, leaving us with system-environment state structure in the original laboratory frame and in the environment $E_1$ frame respectively:
\begin{align}
\rho_{SE_{1}\cdots E_{N}}^{\left(C\right)} & =\sum_{i}p_{i}\ket{\psi_{i}^{S}}\bra{\psi_{i}^{S}}_{S}\otimes\bigotimes_{j=1}^{N}\ket{x_{E_{j}|i}}\bra{x_{E_{j}|i}}_{E_{j}}.\\
\rho_{SCE_{2}\cdots E_{N}}^{\left(E_{1}\right)} & =\sum_{i}p_{i}\ket{\tilde{\psi}_{i}^{S}}\bra{\tilde{\psi}_{i}^{S}} \otimes\ket{-x_{E_{1}|i}}\bra{-x_{E_{1}|i}}_{C}\nonumber \\
 & \quad\otimes\bigotimes_{j=2}^{N}\ket{x_{E_{j}|i}-x_{E_{1}|i}}\bra{x_{E_{j}|i}-x_{E_{1}|i}}_{E_{j}},\\
\ket{\tilde{\psi}_{i,(j=1)}^{S}} & \coloneqq\sum_{q_{S}}\psi\left(q_{S}+x_{E_{j=1}|i}\Big|i\right)\ket{q_{S}}_{S}.
\end{align}
The conditional states already have strong independence, and so the final requirement of SBS objectivity is perfect distinguishability:
\begin{align}
\braket{x_{E_{j}|i}-x_{E_{k}|i}|x_{E_{j}|i^{\prime}}-x_{E_{k}|i^{\prime}}} & =0, \quad\forall i\neq i^{\prime}, j\neq k,\\
\braket{\tilde{\psi}_{i,j}^{S}|\tilde{\psi}_{i^{\prime},j}^{S}} & =0, \quad\forall i\neq i^{\prime}.
\end{align}
We should already have that $\braket{\psi_{i}^{S}|\psi_{i^{\prime}}^{S}}=0,$
$\braket{x_{E_{j}|i}|x_{E_{j}|i^{\prime}}}=0$ $\forall$ $i\neq i^{\prime}$
as this is necessary for the state in the lab frame to be objective.
\end{proof}

\subsection{Proof of Proposition \ref{prop:A-discrete-SBS_with_new_information}\label{app:proof_proposition}}

Here, we prove that a particular SBS state structure---with relevant distinguishability conditions---will be objective in all laboratory and environment frames, albeit with different objective information.

\begin{proof}
Recall the general reduced state of the system in the frame $E_{1}$, Eq.~(\ref{eq:general_reduced_system_state_discrete}). If we now have a new objective information, then there will be some eigendecomposition of the system with new probabilities $\tilde{p}_{\tilde{i}}$:
\begin{equation}
\rho_{S}^{\left(E_{1}\right)}=\sum_{\tilde{i}}\tilde{p}_{\tilde{i}}\ket{\tilde{\psi}_{\tilde{i}}}\bra{\tilde{\psi}_{\tilde{i}}}.
\end{equation}
The reduced system state is
\begin{align}
\rho_{S}^{\left(E_{1}\right)} & =\sum_{i}p_{i}\sum_{q_{C}}t\left(-q_{C},-q_{C}|i,j=1\right)\nonumber \\
 & \quad\times\sum_{q_{S},q_{S}^{\prime}}\psi\left(q_{S}-q_{C}|i\right)\psi^{*}\left(q_{S}^{\prime}-q_{C}|i\right)\ket{q_{S}}\bra{q_{S}^{\prime}}_{S}.
\end{align}
Note that $j=1$ is fixed, while $i$ is not fixed. Recall how in the situation when we had the same objective information, the coefficient $t(-q_{C},-q_{C}|i,j=1)$ was a dirac delta. However, as we are not requiring the \emph{same} objective information here, we can relax this condition.

Instead, we can identify $\sum_{q_{S}}\psi\left(q_{S}-q_{C}|i\right)\ket{q_{S}}_{S}\eqqcolon\ket{\tilde{\psi}_{\tilde{i}}}$ as the conditionally pure states, and impose that they are orthogonal, such that the system decomposition is 
\begin{align}
\rho_{S}^{\left(E_{1}\right)} & =\sum_{i,q_{C|i}\in C_{E_{1}|i}}\tilde{p}\left(i,q_{C|i}\right)\ket{\tilde{\psi}_{\left(i,q_{C|i}\right)}}\bra{\tilde{\psi}_{\left(i,q_{C|i}\right)}}\\
\tilde{p}\left(i,q_{C|i}\right) & =p_{i}t\left(-q_{C|i},-q_{C|i}|i,j=1\right)\\
\ket{\tilde{\psi}_{\left(i,q_{C|i}\right)}} & =\sum_{q_{S}}\psi\left(q_{S}-q_{C|i}|i\right)\ket{q_{S}}_{S},
\end{align}
with new objective information $\left\{ p_{i}t(-q_{C|i},-q_{C|i}|i,j=1)\right\} _{\left(i,q_{C|i}\right)}$, and where $q_{C}=q_{C|i}^{(E_1)}$ has a dependence on $i$ and on the environment frame (here, $E_1$). Hence, the sets $C_{E_1|i}$ describe the possible values $q_C$ can take given index $i$ and environment frame $E_1$. The conditional system states need to be distinguishable for this to be a valid decomposition:
\begin{align}
\braket{\tilde{\psi}_{(i,q_{C|i})}|\tilde{\psi}_{(i^{\prime},q_{C|i}^{\prime})}} & =0,\quad\forall\left(i,q_{C|i}\right)\neq(i^{\prime},q_{C|i}^{\prime}).
\end{align}
Returning to the full system-environment-lab state in frame $E_{1}$, we now have that 
\begin{align}
 & \rho_{SCE_{2}\cdots E_{N}}^{\left(E_{1}\right)} =\sum_{i}p_{i}\sum_{q_{C|i},q_{C|i}^{\prime}}t(-q_{C|i},-q_{C|i}^{\prime}|i,j=1)\nonumber \\
 & \qquad \times \ket{\tilde{\psi}_{\left(i,q_{C|i}\right)}}\bra{\tilde{\psi}_{\left(i,q_{C|i}^{\prime}\right)}}\otimes\ket{q_{C|i}}\bra{q_{C|i}^{\prime}}_{C}\nonumber \\
 & \qquad\otimes\bigotimes_{j=2}^{N}\sum_{q_{E_{j}},q_{E_{j}}^{\prime}}t(q_{E_{j}}-q_{C|i},q_{E_{j}}^{\prime}-q_{C|i}^{\prime}|i,j)\ket{q_{E_{j}}}\bra{q_{E_{j}}^{\prime}}_{E_{j}}.
\end{align}
For SBS structure, the system cannot have extra coherences beyond the $\left\{ \ket{\tilde{\psi}_{\left(i,q_{C|i}\right)}}\right\} _{i,q_{C}}$ basis we have established, and the conditional environment states must be distinguishable. Hence, $t(-q_{C|i},-q_{C|i}^{\prime}|i,j)\overset{!}{=}0$ if $q_{C|i}\neq q_{C|i}^{\prime}$, \emph{i.e the environment $E_{1}$ is incoherent in the $x$ basis. }This means that \emph{all} the environments must be incoherent in the $x$ basis in order for this to hold in all environment frames,
\begin{align}
\rho_{SCE_{2}\cdots E_{N}}^{\left(E_{1}\right)} & =\sum_{i}\sum_{q_{C|i}\in C_{E_1|i}}p_{i}t\left(-q_{C|i}|i,j=1\right)\nonumber \\
 & \quad\times \ket{\tilde{\psi}_{\left(i,q_{C|i}\right)}}\bra{\tilde{\psi}_{\left(i,q_{C|i}\right)}}\otimes\ket{q_{C|i}}\bra{q_{C|i}}_{C}\nonumber \\
 & \quad\otimes\bigotimes_{j=2}^{N}\sum_{q_{E_{j}}}t\left(q_{E_{j}}-q_{C|i}|i,j\right)\ket{q_{E_{j}}}\bra{q_{E_{j}}}_{E_{j}},
\end{align}
where we have written $t\left(q,q|i,j\right)=t\left(q|i,j\right)$. Note that $q_{C}=q_{C|i}^{\left(E_{1}\right)}$ is dependent on $i$ and dependent on the original environment $E_{1}$ states: this is encoded in the coefficients $t(-q_{C|i}|i,j=1)$, as well as the sets $C_{E_1|i}$ for extra clarity. 

Lastly, we impose the distinguishability conditions on the environment states:
\begin{align}
\rho_{E_{j}|\left(i,q_{C}\right)}^{\left(E_{1}\right)}\rho_{E_{j}|\left(i^{\prime},q_{C}^{\prime}\right)}^{\left(E_{1}\right)} & =0,\quad\forall\left(i,q_{C}\right)\neq\left(i^{\prime},q_{C}^{\prime}\right).
\end{align}
Note that this implies any specific $q_{C}$ value is assigned only to one unique index $i$, since we must have $\ket{q_{C|i}}$ being distinguishable for each nonzero combination of $\left(i,q_{C|i}\right)$. In other words, knowing $q_{C}$ automatically gives knowledge of $i$, like a many-to-one function. (Equally, the sets $\{C_{E_1|i}\}_i$ are disjoint across $i$.)
\end{proof}

\section{Detailed discussion of numerical results of dynamic system and two environments\label{app:dynamic}}

In this Appendix we provide a detailed discussion of numerical experiments of~Sec.~\ref{sec:Numerical-simulations}. Subsec.~\ref{app:subsec:measurement_limit_Hamiltonian} considers when the measurement-limit Hamiltonian is the only interaction present. Subsec. \ref{app:subsec:self-evolution_environments} considers additional self-evolution in the environment. Subsec. \ref{app:subsec:mutual_information_transfer} considers environment-environment interactions (without self-evolution). 
Finally, Subsec. \ref{app:subsec:global_evolution} considers the effect of adding a more general global interaction.

\begin{figure}
	\centering
	\includegraphics[width=0.95\linewidth]{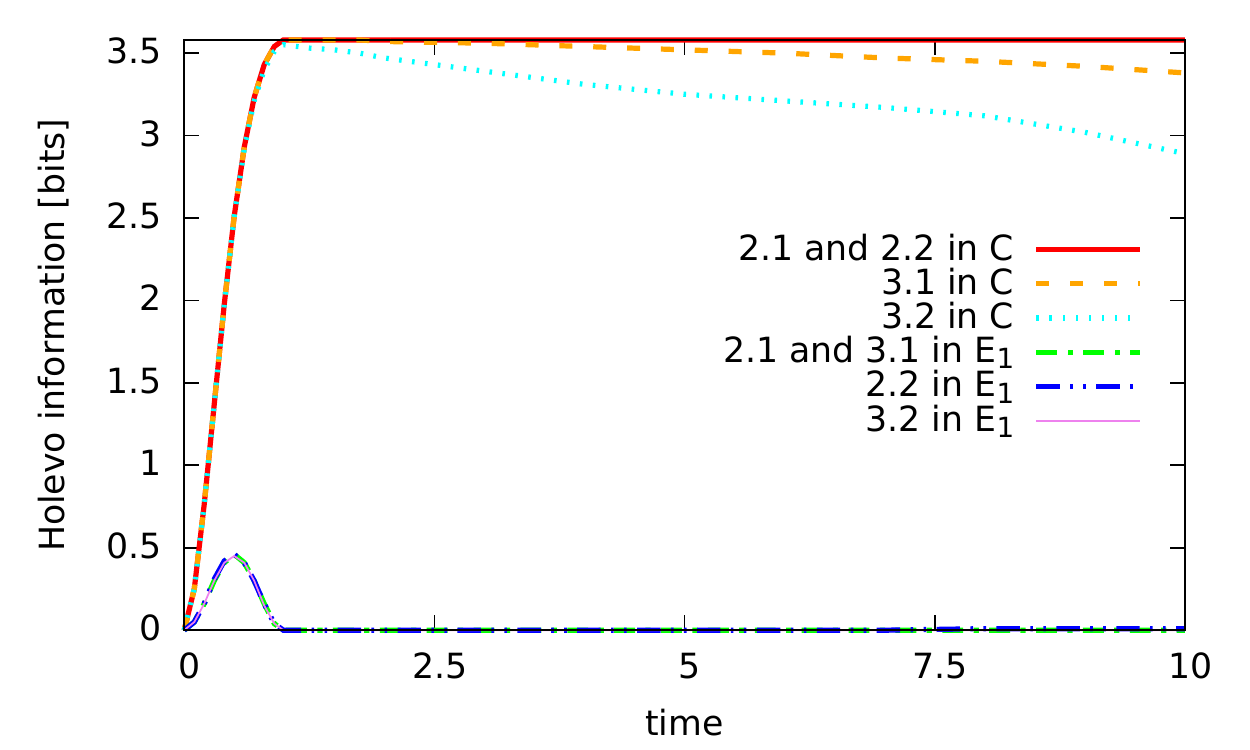}
	\caption{(Color on-line) The Holevo information between central system $S$ and subsystem $E_1$ ($C$) in the reference frame of $C$ ($E_1$). Compared with Fig.~\ref{fig:QMI_S_E1}, the Holevo information is very similar to the quantum mutual information, aside from case 3.2 in frame $E_1$.}
	\label{fig:holevo}
\end{figure}

\subsection{Measurement limit Hamiltonian\label{app:subsec:measurement_limit_Hamiltonian}}

We first consider the case when only the central interaction is present, as described in the main text. We investigate how the system-to-environment information transfer varies given different initial states (Cases 1.1 to 1.5), and we see how the SBS form is preserved when we transform from the reference frame of the external observer, $C$, to the reference frame of the first environment, $E_1$.

Note that this interaction Hamiltonian acts only till time~$t=1$, and there is no further evolution afterwards.

When there is only the central interaction, we find that in frame $C$,  all probabilities remain unchanged, and there is strong independence at all times, except for Case 1.3 (where the initial environment states are entangled).

\subsubsection*{Case 1.1 - information broadcasting with pure environments}

We start with a scenario when the initial state is given by Eq.~\eqref{eq:rho_mpp}, \emph{i.e.} when both environments are pure in state $\ket{0}$, and thus has the largest capacity to obtain information about the central system $S$. In both reference frames $C$ and $E_1$, the states are fully SBS---though note that in the frame of $E_1$, this SBS is \textit{degenerated} and thus trivial, since it has only one non-zero probability, $p_0^{(E_1)} = 1$ at time~$t=0$.

\paragraph{External laboratory perspective.}

The numerical calculations showed that in frame $C$, all conditional environments states, $E_1$ and $E_2$, are initially indistinguishable. As the system interacts with the environments, the fidelity of all conditional states of environments decreases monotonically to $0$ at the time $t=1$, meaning that the conditional states are now distinguishable. Since both environments are symmetric, the fidelity is identical for both of them.

\paragraph{First environment's perspective.}

Fig.~\ref{fig:plotprobpie1t1} shows how all probabilities different from $p_0$ converge monotonically to $0$ until time~$t=1$, corresponding to a trivial objectivity.

The plot of conditional environment-environment mutual information $I_\text{mean}^{(E_1)}$ is show in~Fig.~\ref{fig:it1} in the main text, where a high value of $I_\text{mean}^{(E_1)}$ denotes \emph{lack} of strong independence. In frame $E_1$, the subsystems $C^{(E_1)}$ and $E_2^{(E_1)}$ initially gain correlations amongst themselves up to  time $t=0.5$, before decaying back to strong-independent state by time $t=1$.

The fidelity of individual conditional states of $C$ remains large, but less than $1$. Despite this, the upper bound Eq.~\eqref{eq:errBound} to the distinguishability error is small, and is zero at the time moment~$t=1$. The fidelity of individual conditional states of $E_2$ remains small, but greater than $0$. Comparing the upper bounds Eq.~\eqref{eq:errBound} over both subsystems $C$ and $E_2$ we see that the bound is slightly smaller for the latter.

\subsubsection*{Case 1.2 - information broadcasting with slightly blurred environments}

The second case we consider deals with environmental states that are not pure, but has a small admixture of positions other than the dominant one, given in Eq.~\eqref{eq:rho_mbb}.

\paragraph{External laboratory perspective.}

In this case the final fidelity $B_{E_1}^{(C)}$ is non-zero between the states close to each other, and near to zero for the others. For example at the time $t=0.999$ when the information propagation process is almost done $B_{E_1}^{(C)}(0,1) \approx 0.566$ whereas for $B_{E_1}^{(C)}(0,2) \approx 0.101$ and $B_{E_1}^{(C)}(0,3) \approx 0.001$.

\paragraph{First environment's perspective.}

The probability distribution tends towards distribution corresponding to the initial distribution of the environment, \emph{i.e.} $p_0 = 0.8$, $p_1 = p_{D-1} = 0.1$, c.f. Eq.~\eqref{eq:rho_mbb}. The exact shape of the function is shown in~Fig.~\ref{fig:plotprobpie1t1}.

Again, mutual information is slightly higher for the central system's state $0$. The value of $I_\text{mean}$ is again the largest at the time $t=0.5$ and converges to $0$ at the time moment~$t=1$ for all Cases 1. The plot of $I_\text{mean}^{(E_1)}$ is shown in~Fig.~\ref{fig:it1}.

The fidelity of conditional states of both $C$ and $E_2$ remain non-zero all the time, but fidelity in the case of $E_2$ is smaller. At the time moment~$t=1$, fidelity of states $0$, $1$ and $D-1$ is equal to $1$. The average fidelity for $C$ is at the time $0.999$ equal to $0.786$, and for $E_2$ is equal to $0.122$. Note that for the time moment~$t=1$ the probabilities different than $0$, $1$ and $D-1$ are equal to $0$, and thus $\rho_{C|i}^{(E_1)}$ and $\rho_{E_2|i}^{(E_1)}$ for other values of $i$ are undefined.

\subsubsection*{Case 1.3 - information broadcasting with maximally entangled environments}

The Case 1.3 investigates the initial joint state with maximally entangled environments, given in  Eq.~\eqref{eq:rho_mEE}.

\paragraph{External laboratory perspective.}

Here, the mutual information is constant in time, and equal to $7.170$, such as it results from the maximally entangled state in the dimension under consideration, $D = 12$. The fidelity of conditional states is also constant in time and equal to $1$ between all states. This shows that there is no information transfer in this scenario.

\paragraph{First environment's perspective.}

All probabilities $p_i^{(E_1)}$ are identical, equal to $\frac{1}{D} \approx 0.083$ and constant over time. The fidelity all the time equal to $1$ between all conditional states.

The mutual information is identical for all central system's states. At the time moments $t=0, 1$, the mutual information  is equal to $0$, while the maximum of $1.423$ is attained at the time moment~$t=0.5$. We note that $I_\text{mean}^{(E_1)}$ behaves similarly like in the Case 1.1, but has different values. The plot of $I_\text{mean}^{(E_1)}$ is shown in~Fig.~\ref{fig:it1}.

\subsubsection*{Case 1.4 - information broadcasting with first environment maximally mixed}

In this case the first environment $E_1$ is initially maximally mixed, and the second one pure. One may expect that there is no information transfer to $E_1$.

\paragraph{External laboratory perspective.}

As expected, conditional states on $E_1$, $\{\rho_{E_1|i}^{(C)}\}$ are completely indistinguishable. On the other hand, the conditional states on $E_2$ are fully distinguishable at the time moment~$t=1$.

\paragraph{First environment's perspective.}

From the point of view of $E_1$, all probabilities $p_i^{(E_1)}$ are identical and constant over time. This supports the conclusion that from this perspective no knowledge regarding the central system $S$ is obtained. The mutual information is identical for all states of the central system; it is initially quite large ($3.585$ at the time $t=0$), but monotonically decreasing to $0$ at the time moment~$t=1$. The plot of $I_\text{mean}^{(E_1)}$ is shown in~Fig.~\ref{fig:it1}.

The conditional states of $C$ are completely indistinguishable. The conditional states of $E_2$ are initially indistinguishable, but monotonically tends to the full distinguishability.

\subsubsection*{Case 1.5 - information broadcasting with second environment maximally mixed}

This case is similar to the previous one, where one of the environments is initially maximally mixed, but this time, the second environment, $E_2$ is initially maximaly mixed, as given in Eq.~\eqref{eq:rho_mpm}.

\paragraph{External laboratory perspective.}

As implicated by the symmetry of the scenario, fidelity is identical in values as in the Case 1.4 from $C$'s reference frame, with swapped environments.

\paragraph{First environment's perspective.}

In this case all probabilities $\{p_i^{(E_1)}\}$ are initially identical, and tend to $1$ for $p_0$ and $0$ for the others at the time moment~$t=1$. Furthermore, the values of probabilities are the same as in the Case 1.1 (from $E_1$'s point of view). This means that the different state of the second environment $E_2$ did not influence the perceptions of $E_1$.

The mutual information is $0$ all the time for all central system's states. The fidelity of all conditional states remain quite high for both $C$ and $E_2$ ($0.640$ and $1$, respectively), nonetheless the error probability upper bound decreases to $0$ for both $C$ and $E2$. This is caused not by the distinguishability of conditional states but by the degeneracy of the SBS when only one probability is non-zero.

\subsection{Self-evolution of environments\label{app:subsec:self-evolution_environments}}

In the previous subsection, we dealt with the situation when the only time evolution was caused by the interaction of the environments with the central system. In the second group of cases we add a self-evolution Hamiltonian to each  environment. As described in the main text, the self-evolution Hamiltonians allows for jumps between neighbouring states in the Case 2.1, and for random jumps with random rates in the Case 2.2. For both cases the initial joint state $\rho_\text{mpp}$ is given in Eq.~\eqref{eq:rho_mpp}.

\subsubsection*{Case 2.1 - self-evolution of environments}

\paragraph{External laboratory perspective.}

From symmetry it follows both environments $E_1$ and $E_2$ are identical. The value of mean mutual information, $I_\text{mean}^{(C)}=0$ for all moments of time. The fidelities decreases monotonically from the value $1$ at the time moment~$t=0$ to the value $0$ at the time~$t=1$. Afterwards the fidelities remain constantly $0$. The last observation is a direct consequence of the fact that the fidelity is invariant under unitary rotations, such like the self-rotations of the environments.

\paragraph{First environment's perspective.}

Initially in the $E_1$ reference frame, all probabilities $\{p_i^{(E_1)}\}$ are identical and equal to $\frac{1}{D}$, and tend to $1$ for $p_0^{(E_1)}$ at the time moment~$t=1$. However, for the times of order of $10$, the probabilities adjacent to $p_0^{(E_1)}$ (\emph{i.e.} $\{p_1^{(E_1)}\}$ and $\{p_{D-1}^{(E_1)}\}$) begin to be non-negligible. At the times of order $100$, each probability becomes significant, and $p_0^{(E_1)}$ drops to the value $0.05$ (and will revive). We can see this in Fig.~\ref{fig:plotprobp1e1tm}. With long-term averaging all probabilities are the same.

The conditional mutual information,
\begin{equation}
\label{eq:conditionalMutualInformation}
H_2\left(\rho_{E_1|i}^{(C)}\right) + H_2\left(\rho_{E_2|i}^{(C)}\right) - H_2\left(\rho_{E_1 E_2|i}^{(C)}\right),
\end{equation}
is the largest for the state $i=0$ of the central system. The mean mutual information $I_\text{mean}^{(E_1)}$ is zero at the beginning and at the time~$t=1$, while its maximum of $1.690$ occurs at time~$t=0.5$. After the time moment~$t\geq1$, the error upper bound is constantly equal to $0$. The difference between $I_\text{mean}^{(E_1)}$ in this case and in the Case 1.1 do not exceed $0.0016$ (cf.~Fig.~\ref{fig:it1}).

The fidelities of $C$ are initially equal to $1$. The upper bound of the guessing probability in Eq.~\eqref{eq:errBound}  drops to a small value of $0.04$ at the time moments close to $1$, but for large times it is high and becomes trivial at the time $t\approx20$. On the other hand, for $E_2$ fidelities are also initially equal to $1$, but afterwards they fall quite quickly to the value of $0$ at the time moments close to $1$ and remain equal to $0$ for all time moments after $1$. We note that before time moment~$30$ most of the probabilities $\{p_i^{(E_1)}\}$ are close to $0$ and their fidelity does not influence the upper bound, but after this moment their fidelities are well defined and equal to $0$.

\subsubsection*{Case 2.2 - random self-evolution of environments}

Since in this case the jumps are not distance-dependent, it is expected that locality-related phenomena does not occur in this case, in contrast to the previous Case 2.1.

\paragraph{External laboratory perspective.}

In the $C$'s perspective, the mutual information is always equal to $0$. The fidelities of both $E_1$ and $E_2$ are very similar (up to random fluctuations), and are at a similar level as Case 2.1 (from perspective $C$), but does not converge to $0$ at the time~$t=1$, but instead  to low values between $0.001 - 0.004$ (depending on the choice of the pair of conditional states). For the same reason as in the previous case, they are constant after the time~$t\geq 1$.

\paragraph{First environment's perspective.}

From the point of view of $E_1$ all the probabilities initially identical. $p_0^{(E_1)}$ tend to $1$ for the time $t=1$. Comparing to the Case 2.1 (from $E_1$ perspective), other probabilities becomes non-negligible faster, already at the time of order $5$, see~Fig.~\ref{fig:plotprobp1e1tm}. Another difference is that, as expected, all of these probabilities behave similarly, unlike the local spread to neighbours in the Case 2.1.

The mean mutual information starts with the value $0$, with a local maximum $1.693$ at the time $t=0.5$, and drops back down to $0$ at the time moment~$t=1$. Afterwards it increases and reaches saturation level of about $2.5$ at the time $t\approx 150$---see~Fig.~\ref{fig:ie1tm} in the main text.

Subsystem $C$'s fidelities remains very high all the time, although for the time moments close to $1$ the upper bound on the probability of error in Eq.~\eqref{eq:errBound} is low ($0.109$) but it increases quickly (and becomes trivial for the time $t=8$). For subsystem $E_2$, the fidelities and error probability upper bound drops to $0$ for the time close to $1$ and begin to increase from time $t=10$, to become trivial at the time of order $30$.

\subsection{Mutual information transfer between environments\label{app:subsec:mutual_information_transfer}}

This group of cases investigates interactions between environments (without self-evolution). Case 3.1 has distance dependent interactions (as described in the main text), and Case 3.2 has random rates of jumps. The initial states are such that  both $E_1$ and $E_2$ located at the same position, given in Eq.~\eqref{eq:rho_mpp}.

\subsubsection*{Case 3.1 - simple mutual information transfer}

\paragraph{External laboratory perspective.}

In Fig.~\ref{fig:ictm}, we see that  the mutual information is initially $0$ and gradually increases to significant values (above $1$) for the time of order $20$, then fluctuates, sometimes reaching even the value $3.8$.

\begin{figure}
	\centering
	\includegraphics[width=0.95\linewidth]{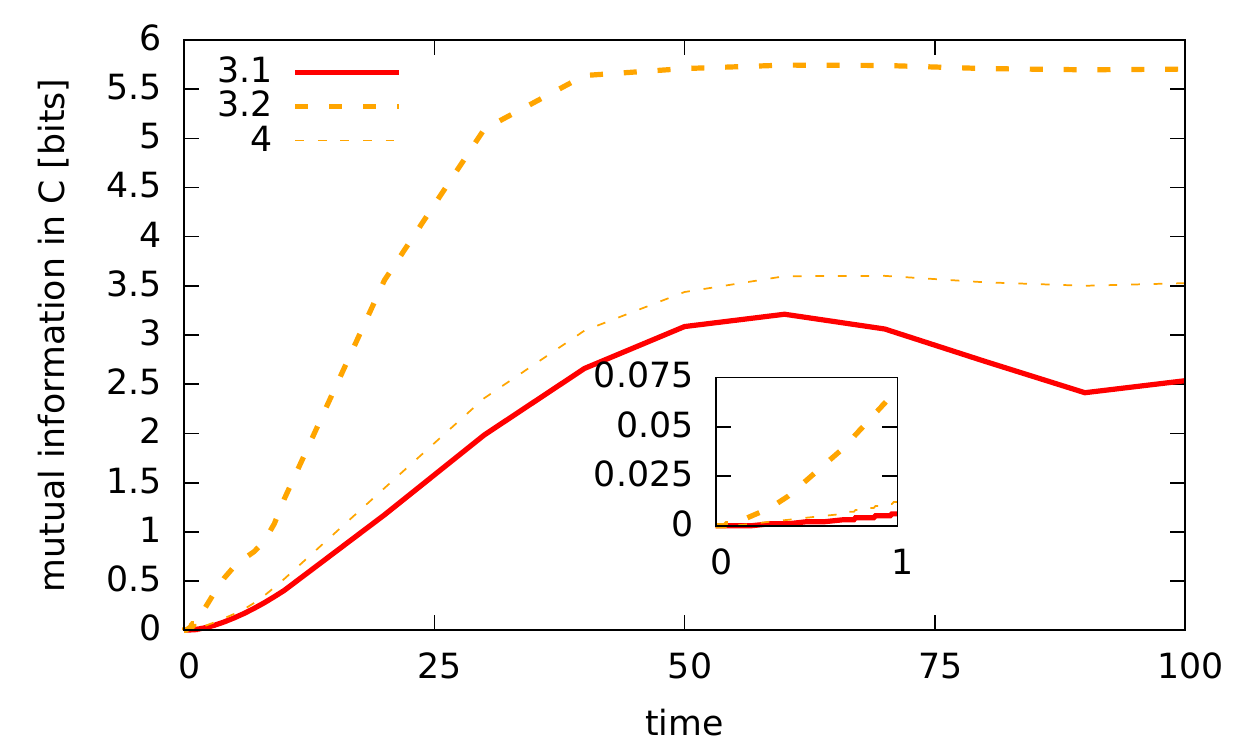}
	\caption{(Color on-line) The values of mean mutual information between environments $I_\text{mean}^{(C)}$ as a function of time for cases from groups 3 and 4 (cf. Tab.~\ref{tab:MI_C}.)}
	\label{fig:ictm}
\end{figure}

As expected from the symmetry, the fidelities are identical between the two environments, and are the greater the closer the states are. The upper bound in Eq.~\eqref{eq:errBound} decreases to a value close to $0$ ($0.05$ at minimum) at the time moments close to $1$, afterwards it increases quickly and becomes trivial at the time close to $20$ and above.

\paragraph{First environment's perspective.}

Here all probabilities initially identical, for the moment~$1$ $p_0^{(E_1)} = 1$. Afterwards at the time of order $2$ the share of $p_1^{(E_1)}$ and $p_{D-1}^{(E_1)}$ (the locally neighbouring states) gradually increases. At the time of the order $40$ some share of $p_2^{(E_1)}$ and $p_{D-2}^{(E_1)}$ (more distant, but still close states) becomes important. The remaining probabilities even for long times (of order $1000000$) stay close to $0$. Detailed plots of these probabilities is shown in Fig.~\ref{fig:plotprobp1e1tm} and Fig.~\ref{fig:plot31pit100}.

\begin{figure}
	\centering
	\includegraphics[width=0.95\linewidth]{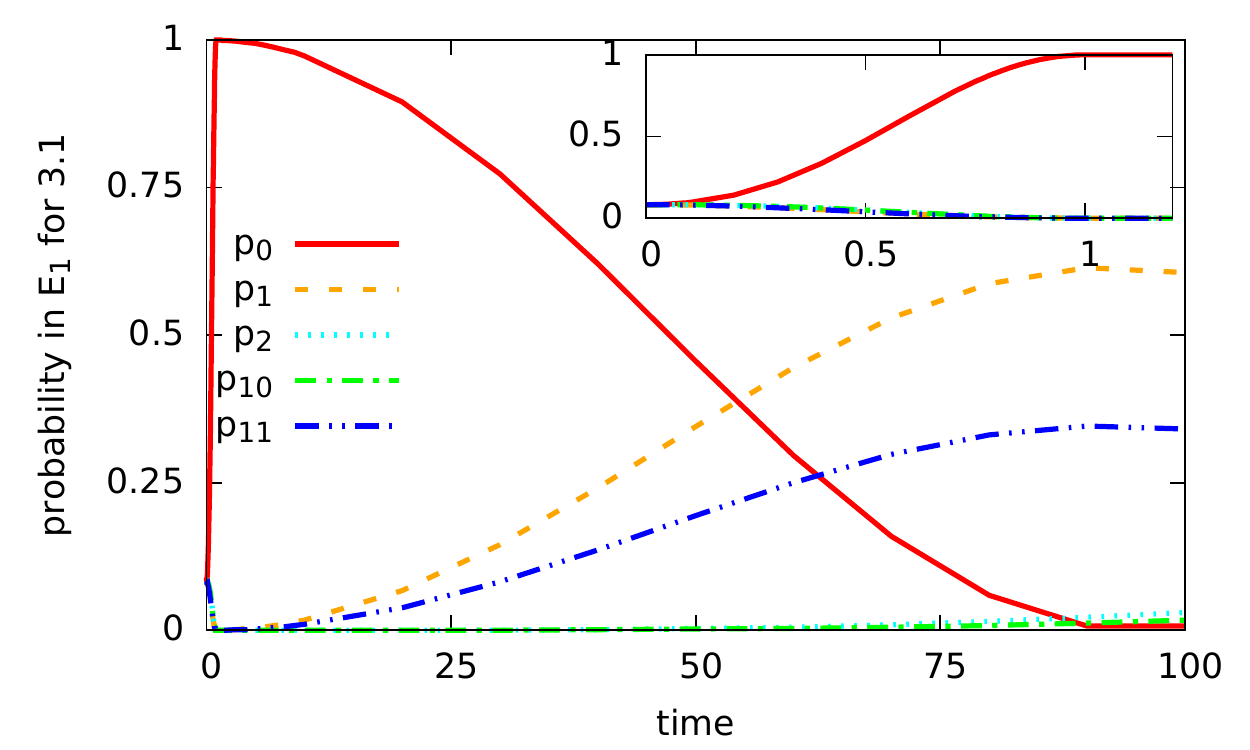}
	\caption{(Color on-line) Probabilities $\{p_i^{(E_1)}\}_{i=0,1,2,10,11}$ for the Case 3.1. $p_0^{(E_1)}$ is dominating at the time~$t=1$, afterwards the probabilities neighbouring at the $12$-dimensional ring, \emph{i.e.} $p_1^{(E_1)}$ and $p_{11}^{(E_1)}$ starts to be non-negligible. After some time also more distant probabilities, $p_2^{(E_1)}$ and $p_{10}^{(E_1)}$, start increasing.}
	\label{fig:plot31pit100}
\end{figure}

Mean mutual information starts with the value $0$, reaches its maximum of $1.689$ at the time $t=0.5$, drops again to $0$ for the time $t\approx 1$ and afterwards remains $0$. The  $I_\text{mean}^{(E_1)}$ in this case and in Case 1.1 are very similar, and never differ more than $0.0009$ (cf.~Fig.~\ref{fig:it1}).

The upper bound on the probability of error  Eq.~\eqref{eq:errBound} drops to a value close to $0$ for the time of order $1$ ($0.051$ and $0.011$ for $C$ and $E_2$, respectively). For the subsystem $C$ it increases quickly, and becomes trivial at the time of order $20$. For the subsystem $E_2$ remains low even for large times (average $0.016$ for a time interval of up to $1000000$).

\subsubsection*{Case 3.2 - mutual information transfer with random Hamiltonian}

\paragraph{External laboratory perspective.}

Mean mutual information is initially $0$ and gradually increases to significant values for the time of order $10$, then reaches saturation and maintains a high constant level of about $5.7$ starting from time about $50$ (and does not drop even for the time of order $1000000$) (see~Fig.~\ref{fig:ictm}).

The error probability upper bound is very high, and thus upper bound is only useful for times close to $1$ (when it is equal to $0.383$), and the bound becomes trivial after time~$t=3$.

\paragraph{First environment's perspective.}

At time $t=0$, all probabilities $p_i^{(E_1)}$ are equal. The probability $p_0^{(E_1)}$ dominates for times close to $1$, then all other probabilities increase to a similar degree. After time of order $50$, all probabilities become very similar (see~Fig.~\ref{fig:plotprobp1e1tm}).

The mean mutual information is initially $0$, the maximum of $1.695$ is attained at the time $t=0.5$, returns close to $0$ for at times close to $1$, gradually increases to saturation level of $2.8$ from time of order $50$ and does not drop (see~Fig.~\ref{fig:ie1tm}).

For $C$'s conditional states the upper bound for error probability Eq.~\eqref{eq:errBound} is very high all the time (giving a non-trivial bound only for the time between $0.9$ and $3$), and minimal (equal to $0.429$) at the time~$t=1$. For $E_2$'s conditional states the bound is also low ($0.136$) only near the moment~$1$, and becomes trivial after time~$t=9$. The fidelity states of both subsystems is high all the time.

\subsection{Slightly disturbed global evolution\label{app:subsec:global_evolution}}

Now, we analyse a case where both self-evolution and environmental interactions are present, with the Hamiltonians are described in the main text (paragraphs following Eq.~\eqref{eq:general_interaction_H}), in conjunction with distortion from a random Hamiltonian over the tripartite state of $C$, $E_1$ and $E_2$. This is the only case where the probabilities $\{p_i^{(C)}\}_{i=0}^{D-1}$ evolve in time and the decoherence factors,
\begin{subequations}
	\begin{equation}
		\Gamma^{(C)} \coloneqq \sum_{i \neq j = 0}^{D-1} \Ab{\bra{i}_S \Tr_{E_1 E_2} \left( \rho_{S E_1 E_2}^{(C)} \right) \ket{j}_S},
	\end{equation}
	\begin{equation}
		\Gamma^{(E_1)} \coloneqq \sum_{i \neq j = 0}^{D-1} \Ab{\bra{i}_S \Tr_{C E_2} \left( \rho_{S C E_2}^{(E_1)} \right) \ket{j}_S},
	\end{equation}
\end{subequations}
have to be taken into account.

\subsubsection*{Case 4 - global evolution with simple self-evolution and mutual information transfer}

\paragraph{External laboratory perspective.}

The fluctuations from the uniform probability distribution of $\{p_i^{(C)}\}_{i=0}^{D-1}$ are only visible from time~$t=50$, but nevertheless the standard deviation of the probabilities at any given time (at least till time $t=1000000$) does not exceed $0.002$.

The mean mutual information between the conditional environments increases from $0$, by the value $0.012$ for time $t=1$, to the value of saturation equal to about $3.5$, and after time $t=50$ it fluctuates around this value, see~Fig.~\ref{fig:ictm}. Hence, the environments do not have strong independence.

The upper bound in Eq.~\eqref{eq:errBound} is large for both $E_1$ and $E_2$, the bound is reasonable only close to time~$t=1$, and becomes trivial before time $t=20$. Fidelities remains low till time around $t=10$. The long time average of $\Gamma^{(E_1)}$ is $0.245$, its value at the time~$t=1$ is $0.003$.

\paragraph{First environment's perspective.}

In Fig.~\ref{fig:plotprobp1e1tm}, the distribution of $\{p_i^{(E_1)}\}_{i=0}^{D-1}$ is initially uniform, and by time~$t=1$ it becomes dominated by $p_0^{(E_1)} \approx 0.999$. The adjacent probabilities $p_1^{(E_1)}$ and $p_{D-1}^{(E_1)}$ begin to increase and become significant for the time $t\approx 4$. From time $t\approx300$ the distribution is again uniform with almost no fluctuation (the standard deviation is $0.002$).

The mean mutual information starts from the value $0$, the local maximum of $1.692$ is attained for the time $t=0.5$, then it decreases to $0.001$ at the time~$t=1$, and afterwards it reaches saturation about $3.5$ for the time $t\approx 300$ and remains at this level. The plot of $I_\text{mean}^{(E_1)}$ is shown in~Fig.~\ref{fig:ie1tm}. Hence, in the long time regime, there is no strong independence.

The error probability upper bound is low only around time~$t=1$ ($0.142$ for $C$ and $0.048$ for $E_2$), and the fidelities are high for both environments, meaning that the conditional states are not very distinguishable. The long time average of $\Gamma^{(E_1)}$ is $0.235$, its value at the time moment~$1$ is $0.015$.

\phantomsection
\addcontentsline{toc}{section}{References}

\bibliographystyle{apsrev4-1}
\bibliography{biblio}

\end{document}